\documentclass[journal]{IEEEtran}

\usepackage[utf8]{inputenc}
\usepackage{amsmath}
\usepackage{amssymb}
\usepackage{graphicx}
\usepackage{graphics}
\usepackage{color, xcolor}
\usepackage{psfrag}
\usepackage{cite}
\usepackage{balance}
\usepackage{algorithm}
\usepackage{accents}
\usepackage{amsthm}
\usepackage{bm}
\usepackage{url}
\usepackage{algorithmic}
\usepackage[english]{babel}
\usepackage{multirow}
\usepackage{enumerate}
\usepackage{cases}
\usepackage{stfloats}
\usepackage{dsfont}
\usepackage{soul}
\usepackage{amsfonts}
\usepackage{fancyhdr}
\usepackage{hhline}
\usepackage{array}
\usepackage{booktabs}
\usepackage{framed}
\usepackage{float}
\usepackage{threeparttable}
\usepackage[caption=false,font=footnotesize]{subfig}
\addto\captionsenglish{}

\usepackage{hyperref}
\hypersetup{colorlinks=false,
            linkcolor=red,
            anchorcolor=green,
            citecolor=black}

\begin{document}

\newtheorem{theorem}{Theorem}
\newtheorem{acknowledgement}[theorem]{Acknowledgement}
\newtheorem{axiom}[theorem]{Axiom}
\newtheorem{case}[theorem]{Case}
\newtheorem{claim}[theorem]{Claim}
\newtheorem{conclusion}[theorem]{Conclusion}
\newtheorem{condition}[theorem]{Condition}
\newtheorem{conjecture}[theorem]{Conjecture}
\newtheorem{criterion}[theorem]{Criterion}
\newtheorem{definition}{Definition}
\newtheorem{exercise}[theorem]{Exercise}
\newtheorem{lemma}{Lemma}
\newtheorem{corollary}{Corollary}
\newtheorem{notation}[theorem]{Notation}
\newtheorem{problem}[theorem]{Problem}
\newtheorem{proposition}{Proposition}
\newtheorem{solution}[theorem]{Solution}
\newtheorem{summary}[theorem]{Summary}
\newtheorem{assumption}{Assumption}
\newtheorem{example}{\bf Example}
\newtheorem{remark}{\bf Remark}

\newtheorem{thm}{Corollary}[section]
\renewcommand{\thethm}{\arabic{section}.\arabic{thm}}

\def\qed{$\Box$}
\def\QED{\mbox{\phantom{m}}\nolinebreak\hfill$\,\Box$}
\def\proof{\noindent{\emph{Proof:} }}
\def\poof{\noindent{\emph{Sketch of Proof:} }}
\def
\endproof{\hspace*{\fill}~\qed
\par
\endtrivlist\unskip}
\def\endproof{\hspace*{\fill}~\qed\par\endtrivlist\vskip3pt}

\def\E{\mathsf{E}}
\def\eps{\varepsilon}
\def\phi{\varphi}
\def\Lsp{{\boldsymbol L}}
\def\Bsp{{\boldsymbol B}}
\def\lsp{{\boldsymbol\ell}}
\def\Ltsp{{\Lsp^2}}
\def\Lpsp{{\Lsp^p}}
\def\Linsp{{\Lsp^{\infty}}}
\def\LtR{{\Lsp^2(\Rst)}}
\def\ltZ{{\lsp^2(\Zst)}}
\def\ltsp{{\lsp^2}}
\def\ltZt{{\lsp^2(\Zst^{2})}}
\def\ninN{{n{\in}\Nst}}
\def\oh{{\frac{1}{2}}}
\def\grass{{\cal G}}
\def\ord{{\cal O}}
\def\dist{{d_G}}
\def\conj#1{{\overline#1}}
\def\ntoinf{{n \rightarrow \infty}}
\def\toinf{{\rightarrow \infty}}
\def\tozero{{\rightarrow 0}}
\def\trace{{\operatorname{trace}}}
\def\ord{{\cal O}}
\def\UU{{\cal U}}
\def\rank{{\operatorname{rank}}}
\def\acos{{\operatorname{acos}}}

\def\SINR{\mathsf{SINR}}
\def\SNR{\mathsf{SNR}}
\def\SIR{\mathsf{SIR}}
\def\tSIR{\widetilde{\mathsf{SIR}}}
\def\Ei{\mathsf{Ei}}
\def\l{\left}
\def\r{\right}
\def\lb{\left\{}
\def\rb{\right\}}

\setcounter{page}{1}

\newcommand{\eref}[1]{(\ref{#1})}
\newcommand{\fig}[1]{Fig.\ \ref{#1}}

\def\bydef{:=}
\def\ba{{\mathbf{a}}}
\def\bb{{\mathbf{b}}}
\def\bc{{\mathbf{c}}}
\def\bd{{\mathbf{d}}}
\def\bee{{\mathbf{e}}}
\def\bff{{\mathbf{f}}}
\def\bg{{\mathbf{g}}}
\def\bh{{\mathbf{h}}}
\def\bi{{\mathbf{i}}}
\def\bj{{\mathbf{j}}}
\def\bk{{\mathbf{k}}}
\def\bl{{\mathbf{l}}}
\def\bm{{\mathbf{m}}}
\def\bn{{\mathbf{n}}}
\def\bo{{\mathbf{o}}}
\def\bp{{\mathbf{p}}}
\def\bq{{\mathbf{q}}}
\def\br{{\mathbf{r}}}
\def\bs{{\mathbf{s}}}
\def\bt{{\mathbf{t}}}
\def\bu{{\mathbf{u}}}
\def\bv{{\mathbf{v}}}
\def\bw{{\mathbf{w}}}
\def\bx{{\mathbf{x}}}
\def\by{{\mathbf{y}}}
\def\bz{{\mathbf{z}}}
\def\b0{{\mathbf{0}}}

\def\bA{{\mathbf{A}}}
\def\bB{{\mathbf{B}}}
\def\bC{{\mathbf{C}}}
\def\bD{{\mathbf{D}}}
\def\bE{{\mathbf{E}}}
\def\bF{{\mathbf{F}}}
\def\bG{{\mathbf{G}}}
\def\bH{{\mathbf{H}}}
\def\bI{{\mathbf{I}}}
\def\bJ{{\mathbf{J}}}
\def\bK{{\mathbf{K}}}
\def\bL{{\mathbf{L}}}
\def\bM{{\mathbf{M}}}
\def\bN{{\mathbf{N}}}
\def\bO{{\mathbf{O}}}
\def\bP{{\mathbf{P}}}
\def\bQ{{\mathbf{Q}}}
\def\bR{{\mathbf{R}}}
\def\bS{{\mathbf{S}}}
\def\bT{{\mathbf{T}}}
\def\bU{{\mathbf{U}}}
\def\bV{{\mathbf{V}}}
\def\bW{{\mathbf{W}}}
\def\bX{{\mathbf{X}}}
\def\bY{{\mathbf{Y}}}
\def\bZ{{\mathbf{Z}}}

\def\bxi{{\boldsymbol{\xi}}}

\def\sT{{\mathsf{T}}}
\def\sH{{\mathsf{H}}}
\def\cmp{{\text{cmp}}}
\def\cmm{{\text{cmm}}}
\def\WPT{{\text{WPT}}}
\def\lo{{\text{lo}}}
\def\gl{{\text{gl}}}

\def\tT{{\widetilde{T}}}
\def\tF{{\widetilde{F}}}
\def\tP{{\widetilde{P}}}
\def\tG{{\widetilde{G}}}
\def\tbh{{\widetilde{\mathbf{h}}}}
\def\tbg{{\widetilde{\mathbf{g}}}}

\def\mA{{\mathbb{A}}}
\def\mB{{\mathbb{B}}}
\def\mC{{\mathbb{C}}}
\def\mD{{\mathbb{D}}}
\def\mE{{\mathbb{E}}}
\def\mF{{\mathbb{F}}}
\def\mG{{\mathbb{G}}}
\def\mH{{\mathbb{H}}}
\def\mI{{\mathbb{I}}}
\def\mJ{{\mathbb{J}}}
\def\mK{{\mathbb{K}}}
\def\mL{{\mathbb{L}}}
\def\mM{{\mathbb{M}}}
\def\mN{{\mathbb{N}}}
\def\mO{{\mathbb{O}}}
\def\mP{{\mathbb{P}}}
\def\mQ{{\mathbb{Q}}}
\def\mR{{\mathbb{R}}}
\def\mS{{\mathbb{S}}}
\def\mT{{\mathbb{T}}}
\def\mU{{\mathbb{U}}}
\def\mV{{\mathbb{V}}}
\def\mW{{\mathbb{W}}}
\def\mX{{\mathbb{X}}}
\def\mY{{\mathbb{Y}}}
\def\mZ{{\mathbb{Z}}}

\def\cA{\mathcal{A}}
\def\cB{\mathcal{B}}
\def\cC{\mathcal{C}}
\def\cD{\mathcal{D}}
\def\cE{\mathcal{E}}
\def\cF{\mathcal{F}}
\def\cG{\mathcal{G}}
\def\cH{\mathcal{H}}
\def\cI{\mathcal{I}}
\def\cJ{\mathcal{J}}
\def\cK{\mathcal{K}}
\def\cL{\mathcal{L}}
\def\cM{\mathcal{M}}
\def\cN{\mathcal{N}}
\def\cO{\mathcal{O}}
\def\cP{\mathcal{P}}
\def\cQ{\mathcal{Q}}
\def\cR{\mathcal{R}}
\def\cS{\mathcal{S}}
\def\cT{\mathcal{T}}
\def\cU{\mathcal{U}}
\def\cV{\mathcal{V}}
\def\cW{\mathcal{W}}
\def\cX{\mathcal{X}}
\def\cY{\mathcal{Y}}
\def\cZ{\mathcal{Z}}
\def\cd{\mathcal{d}}
\def\Mt{M_{t}}
\def\Mr{M_{r}}
\def\O{\Omega_{M_{t}}}
\newcommand{\figref}[1]{{Fig.}~\ref{#1}}
\newcommand{\tabref}[1]{{Table}~\ref{#1}}

\newcommand{\fb}{\tx{fb}}
\newcommand{\nf}{\tx{nf}}
\newcommand{\BC}{\tx{(bc)}}
\newcommand{\MAC}{\tx{(mac)}}
\newcommand{\Pout}{p_{\mathsf{out}}}
\newcommand{\nnn}{\nn\\}
\newcommand{\FB}{\tx{FB}}
\newcommand{\TX}{\tx{TX}}
\newcommand{\RX}{\tx{RX}}
\renewcommand{\mod}{\tx{mod}}
\newcommand{\m}[1]{\mathbf{#1}}
\newcommand{\td}[1]{\tilde{#1}}
\newcommand{\sbf}[1]{\scriptsize{\textbf{#1}}}
\newcommand{\stxt}[1]{\scriptsize{\textrm{#1}}}
\newcommand{\suml}[2]{\sum\limits_{#1}^{#2}}
\newcommand{\sumlk}{\sum\limits_{k=0}^{K-1}}
\newcommand{\eqhsp}{\hspace{10 pt}}
\newcommand{\tx}[1]{\texttt{#1}}
\newcommand{\Hz}{\ \tx{Hz}}
\newcommand{\sinc}{\tx{sinc}}
\newcommand{\diag}{\mathrm{diag}}
\newcommand{\MAI}{\tx{MAI}}
\newcommand{\ISI}{\tx{ISI}}
\newcommand{\IBI}{\tx{IBI}}
\newcommand{\CN}{\tx{CN}}
\newcommand{\CP}{\tx{CP}}
\newcommand{\ZP}{\tx{ZP}}
\newcommand{\ZF}{\tx{ZF}}
\newcommand{\SP}{\tx{SP}}
\newcommand{\MMSE}{\tx{MMSE}}
\newcommand{\MINF}{\tx{MINF}}
\newcommand{\RC}{\tx{MP}}
\newcommand{\MBER}{\tx{MBER}}
\newcommand{\MSNR}{\tx{MSNR}}
\newcommand{\MCAP}{\tx{MCAP}}
\newcommand{\vol}{\tx{vol}}
\newcommand{\ah}{\hat{g}}
\newcommand{\tg}{\tilde{g}}
\newcommand{\teta}{\tilde{\eta}}
\newcommand{\heta}{\hat{\eta}}
\newcommand{\uh}{\m{\hat{s}}}
\newcommand{\eh}{\m{\hat{\eta}}}
\newcommand{\hv}{\m{h}}
\newcommand{\hh}{\m{\hat{h}}}
\newcommand{\Po}{P_{\mathrm{out}}}
\newcommand{\Poh}{\hat{P}_{\mathrm{out}}}
\newcommand{\Ph}{\hat{\gamma}}
\newcommand{\mat}[1]{\begin{matrix}#1\end{matrix}}
\newcommand{\ud}{^{\dagger}}
\newcommand{\C}{\mathcal{C}}
\newcommand{\nn}{\nonumber}
\newcommand{\nInf}{U\rightarrow \infty}

\title{In-Memory Computing Enabled Deep MIMO Detection to Support Ultra-Low-Latency Communications}

\markboth{IEEE Transactions on Mobile Computing}%
{Shell \MakeLowercase{\textit{et al.}}: Bare Demo of IEEEtran.cls for IEEE Journals}

\author{Tingyu Ding,~\IEEEmembership{Graduate Student Member,~IEEE}, Qunsong Zeng,~\IEEEmembership{Member,~IEEE}, and  Kaibin Huang,~\IEEEmembership{Fellow,~IEEE}
\thanks{T. Ding, Q. Zeng, and K. Huang are with the Department of Electrical and Computer Engineering, The University of Hong Kong, Hong Kong, China. 
(Equal contribution: T. Ding and Q. Zeng).
Corresponding authors: Q. Zeng and K. Huang (emails: qszeng@eee.hku.hk, huangkb@hku.hk).} 
}

\maketitle

\begin{abstract}
The development of sixth-generation (6G) mobile networks imposes unprecedented latency and reliability demands on multiple-input multiple-output (MIMO) communication systems, a key enabler of high-speed radio access. Recently, deep unfolding-based detectors, which map iterative algorithms onto neural network architectures, have emerged as a promising approach, combining the strengths of model-driven and data-driven methods to achieve high detection accuracy with relatively low complexity. However, algorithmic innovation alone is insufficient; software-hardware co-design is essential to meet the extreme latency requirements of 6G (i.e., 0.1 milliseconds). This motivates us to propose leveraging in-memory computing, which is an analog computing technology that integrates memory and computation within memristor circuits, to perform the intensive matrix-vector multiplication (MVM) operations inherent in deep MIMO detection at the nanosecond scale. Specifically, we introduce a novel architecture, called the deep in-memory MIMO (IM-MIMO) detector, characterized by two key features. First, each of its cascaded computational blocks is decomposed into channel-dependent and channel-independent neural network modules. Such a design minimizes the latency of memristor reprogramming in response to channel variations, which significantly exceeds computation time. Second, we develop a customized detector-training method that exploits prior knowledge of memristor-value statistics to enhance robustness against programming noise. Furthermore, we conduct a comprehensive analysis of the IM-MIMO detector’s performance, evaluating detection accuracy, processing latency, and hardware complexity. Our study quantifies detection error as a function of various factors, including channel noise, memristor programming noise, and neural network size. The insights gained from these tradeoffs provide valuable guidance for practical system implementation.
\end{abstract}
\section{Introduction}
The sixth-generation (6G) mobile networks are designed to support extreme mobility at speeds of up to 1000 km/h, reducing the required air-latency from 1 ms in 5G to just 0.1 ms in 6G \cite{banafaa20236g}. 
At the same time, the increasing demand for higher data rates (up to 100 Gbps) is driving the evolution of multiple-input multiple-output (MIMO) systems toward massive MIMO configurations \cite{wang2024tutorial, bjornson2023twenty, marzetta2010noncooperative}. 
While the expansion of antenna arrays improves spatial multiplexing gains, it also substantially increases the complexity of baseband processing and poses challenges for implementing latency-critical 6G applications \cite{dardari2021holographic, khan2022digital, frotzscher2014requirements, zhai2024golden}. 
The continuous demand for higher computational speeds from 2G to 5G has largely been addressed by scaling down transistor sizes following Moore's Law \cite{shalf2020future}. 
However, as we enter the 6G era, this approach is reaching its physical limits, with transistor sizes approaching atomic scales \cite{leiserson2020there}. 
To ensure continuous scaling of computational power and efficiency, an alternative strategy is to revolutionize computing architectures, which has led to the emergence of in-memory computing (IMC) \cite{verma2019memory}, \cite{wen2024fusion}. 
This paradigm shift motivates this work to design novel baseband processing architectures for MIMO detection, termed in-memory MIMO (IM-MIMO) processing, which aims to meet the ultra-low-latency requirements of 6G systems.

MIMO detection constitutes the dominant computational load in the baseband processing of MIMO systems \cite{albreem2021deep}. The optimal maximum likelihood (ML) detector requires an exhaustive search, with
computational complexity scaling exponentially with MIMO dimensionality \cite{zhu2002performance}. Although sub-optimal algorithms (e.g., approximate message passing \cite{jeon2015optimality} and
expectation propagation \cite{cespedes2014expectation}) have lower complexity, they still require a large number of iterations and are insufficient for real-time applications. To overcome these limitations, data-driven methods have emerged \cite{albreem2021deep}, \cite{tan2018improving}. They leverage deep neural networks (DNNs) to accurately approximate optimal nonlinear detection functions. Their drawbacks are caused by the “black-box” nature of DNNs, resulting in limited interpretability and lack of performance guarantees \cite{balatsoukas2019deep}. This challenge motivates the development of a new approach known as deep unfolding, which unfolds a model-based iterative algorithm into cascaded blocks within a neural network architecture \cite{hershey2014deep, he2020model, khani2020adaptive}. For example, the Detection Network (DetNet) unfolds the iterations of projected
gradient descent (PGD), achieving detection accuracy close to that of the ML detector \cite{samuel2019learning}. Deep unfolding combines the interpretability and structural advantages of classical algorithms with the optimization capabilities of data-driven learning \cite{he2020model}. Nevertheless, deep unfolding-based MIMO detectors still heavily rely on matrix-vector multiplications (MVMs) within network blocks. 
Due to the high computational complexity of these operations, detectors implemented on conventional digital baseband processors still struggle to realize ultra-low-latency large-scale MIMO detection.

The performance of traditional complementary metal-oxide-semiconductor (CMOS)-based digital processors is limited by two factors: transistor scaling \cite{shalf2020future} and the memory wall \cite{kim2023samsung}. The former refers to the physical limitations that hinder further miniaturization of transistors \cite{shalf2020future}. 
The latter arises from frequent data transfers between physically separated processing and memory units, with memory access accounting for approximately 80\% of total latency \cite{shalf2020future}. 
A particularly promising solution to the memory wall is the aforementioned IMC architecture, which performs computation directly within memory units \cite{verma2019memory}. 
A common realization of IMC is based on crossbar arrays of non-volatile memory elements, termed \emph{memristors}, where computation is executed by exploiting the voltage-current relationships governed by Ohm's and Kirchhoff's laws \cite{ielmini2018memory}. 
As a result, this architecture enables highly parallel execution of multiply-and-accumulate (MAC) operations for MVM to achieve ultra-low latency, for example, completing MVM within a single input-pulse duration on the nanosecond scale \cite{wu2021atomically}. This makes IMC particularly well-suited for MVM-intensive tasks such as MIMO detection.

Despite significant advancements in CMOS-based MIMO detectors implemented by application-specific integrated circuits (ASICs), these designs mainly rely on fixed and lightweight arithmetic operations and are inherently constrained by the von Neumann bottleneck~\cite{jeon2019354, prabhu20173, tang20210, lee202540, li2025deep}. While buffers and interconnects can mitigate data movement issues, they cannot eliminate them entirely. For example, \cite{prabhu20173} employs ping-pong buffers and pipeline stages for data reordering, but data must still be transferred between the unified processing element (PE) array and distributed registers. The detector in~\cite{lee202540} reduces memory access via layer ordering, but the large Gram matrix of the orthogonal time-frequency space (OTFS) system still requires frequent reads from on-chip static random-access memory (SRAM). In \cite{li2025deep}, the PE utilization is optimized to 90\% through reconfigurable arrays, but the arbiter's data routing between Gram/matched filter (MF)/interference computation modes introduces additional latency, resulting in an idle phase of 144 clock cycles due to updates of channel states. 
Furthermore, parallelism is limited by CMOS scalability. The grouped layer-parallel architecture in~\cite{tang20210} reduces the PE count by 4.24$\times$ but sacrifices 2.41$\times$ throughput, and the linear MAC array in~\cite{jeon2019354} fails to realize parallel computing, requiring 32 cycles for a $32\times32$ matrix multiplication. 
To overcome the von Neumann bottleneck and the efficiency limitations of conventional CMOS-based detectors, recent studies have explored IMC for MIMO detection.

Researchers have proposed applying memristor-based IMC to achieve single-shot computation of zero-forcing (ZF) and minimum mean squared error (MMSE) MIMO detectors \cite{zeng2023realizing}. Subsequent efforts have focused on the design of nonlinear IM-MIMO detectors, including closed-loop ridge-regression-based detection \cite{mannocci2022analogue} and ML detection \cite{ren2024accelerating}. 
However, unlike digital processing, a critical limitation of IMC is the inaccuracy of programming memristors prior to computation. 
Repeated programming is required to suppress the resultant programming noise and achieve satisfactory accuracy, but this comes at the cost of substantial latency. 
This challenge is particularly relevant for IM-MIMO detection, as time-varying channels require periodic programming of channel matrices into crossbar arrays. 
As observed in the literature, this step often dominates the overall system latency with IM-MIMO detection \cite{zeng2023realizing}; the issue is exacerbated when iterative detection algorithms are employed \cite{ren2024accelerating}. 
Therefore, crossbar-programming noise and latency are key challenges to tackle before IM-MIMO detection can be effective for systems with fast fading. 

In light of prior work, deep IM-MIMO detection based on deep unfolding remains an underexplored area. This approach offers several notable advantages. 
First, matrix operations such as multiplication can be executed at the nanosecond timescale. 
This helps to overcome the computational bottleneck inherent in deep unfolding-based MIMO detectors, which require numerous MVM operations. 
Second, unlike iterative algorithms, the deep unfolding architecture, comprising fixed-depth, feed-forward neural network blocks, is largely static after initial programming. 
This significantly reduces the need for frequent crossbar reprogramming, a major source of latency in IMC systems. 
Third, deep unfolding architectures inherit the inherent robustness of neural networks to external noise, such as parameter quantization noise \cite{hershey2014deep}. 
This robustness can be leveraged to tolerate memristor programming noise, potentially eliminating the costly verification step during programming. 
These advantages motivate the development of a deep IM-MIMO detection framework in this work.

To the best of our knowledge, this work represents the first attempt to design a deep IM-MIMO detector with the objective of overcoming the computational bottleneck of the state-of-the-art (SOTA) deep unfolding-based MIMO detection. The key contributions and findings of this work are summarized as follows:

\begin{itemize}
    \item \textbf{Circuit design of deep IM-MIMO detector:} 
    The proposed architecture, implemented using memristor-based crossbar arrays, addresses critical challenges in processing latency, hardware complexity, and detection accuracy.
    To begin with, all channel-relevant operations across network blocks are consolidated into a unified component termed the channel-dependent module. This module, efficiently realized via crossbar arrays, is the sole component requiring reprogramming in response to channel variations. Owing to its significantly smaller size relative to the full model and its reusability across all blocks, this design substantially reduces both programming latency and hardware complexity.
    Furthermore, a customized training approach is introduced to enhance robustness against programming noise. This strategy incorporates a behavioral model of memristor conductance evolution as a function of programming pulse count, along with empirical characterization of cycle-to-cycle (C2C) variations, i.e., stochastic fluctuations in conductance change between successive programming pulses. During training, synthetic noise following this statistical distribution is injected, enabling the network to learn noise-tolerant representations. As a result, the system gains improved robustness to programming noise during inference.

    \item \textbf{Performance analysis}: We mathematically characterize the performance of the proposed deep IM-MIMO design in terms of detection accuracy, processing latency, and hardware complexity. 
    The analysis quantifies the dependence of detection error on channel noise and programming noise. 
    In particular, the scaling laws of detection error w.r.t. the network size and the number of unfolded blocks are derived.    
    The results suggest that detection accuracy can be improved by reducing programming noise, increasing the network size, or adding more unfolded blocks at the cost of increased processing latency and hardware complexity. 
    On the other hand, the overall processing latency comprises two components: the programming latency for writing the channel matrix and the computation latency for detecting subsequent symbols. The former exhibits a modest super-linear scaling with the MIMO size, whereas the latter scales linearly with the number of unfolded blocks.
    In terms of hardware complexity, both the network size and the block count are key factors. A larger network size necessitates larger crossbar arrays, increasing the number of memristors and peripheral circuits, while a higher block count requires the cascade of more circuit modules. Both increase the overall hardware cost.

    \item \textbf{Simulation results:} The performance of the proposed deep IM-MIMO detector is validated through extensive simulations. 
    It was observed that the proposed deep IM-MIMO detector can achieve detection accuracy close to that of the optimal ML detector, with its accuracy being further enhanced through the integration of channel coding.
    Moreover, the proposed noise-aware training approach exhibits 2.7$\times$ enhanced robustness against C2C variations.
    Compared with SOTA ASIC-based MIMO detectors, the proposed deep IM-MIMO detector achieves the best accuracy performance, while delivering a maximum throughput gain of 159.3$\times$, an energy efficiency gain of 8.1$\times$, and an area efficiency gain of 3.9$\times$. Furthermore, the proposed MIMO detector outperforms conventional digital processors by offering up to $10^3\times$ higher throughput and $1.1\times10^4$ higher energy efficiency.
    
\end{itemize}

The remainder of this paper is organized as follows. The system models and metrics of deep IM-MIMO detection are described in Section~\ref{Sec: PRELIMINARIES}. The architecture and the circuit design details are presented in Section \ref{IMC Architecture}. Section \ref{analysis} provides theoretical performance analysis for our design. In Section \ref{results}, we validate the improved performance of the proposed design with simulation results. Finally, conclusions are drawn in Section \ref{conclusion}.

\section{Preliminaries, Models, and Metrics}\label{Sec: PRELIMINARIES}
To facilitate a better understanding of the proposed deep IM-MIMO detection, this section presents preliminaries on the deep unfolding-based MIMO detector, along with the underlying working principles of IMC. Relevant system models and metrics are described as follows.

\subsection{Deep Unfolding-Based MIMO Detection}
We consider a MIMO system with \textit{$N_t$} transmit antennas and \textit{$N_r$} receive antennas. Let $\mathcal{\tilde S}$ be the finite set of constellation points, and $\mathbf{\tilde{x}}\in \mathcal{\tilde S}^{N_t}$ denotes the transmitted symbol vector. The received signal vector $\mathbf{\tilde{y}}\in \mathbb{C}^{N_r}$ can be expressed as
\begin{equation}
    \mathbf{\tilde y} = \mathbf{\tilde H}\mathbf{\tilde x} + \mathbf{\tilde n}, \label{complex_equation}
\end{equation}
where $\mathbf{\tilde H} \in \mathbb{C}^{{N_r\times N_t}}$ is the channel matrix with elements following i.i.d. Rayleigh fading, yielding $\mathbf{\tilde H} \sim \mathcal{CN}(\mathbf{0},2\mathbf{I})$, and $\mathbf{\tilde{n}}\in \mathbb{C}^{N_r}$ represents the additive white Gaussian noise (AWGN) following ${\mathbf{\tilde n}}\sim \mathcal{CN}(\mathbf{0},2\sigma _{{n}}^2\mathbf{I})$. The transmitted symbol vector is normalized to have unit power, i.e., $\|\mathbf{\tilde x}\|_2^2=1$. The task of MIMO detection is to estimate $\mathbf{\tilde{x}}$ based on the knowledge of $\mathbf{\tilde y}$, $\mathbf{\tilde H}$, and the statistics of noise $\mathbf{\tilde{n}}$.

To process complex-valued signals, we reformulate (\ref{complex_equation}) by applying a complex-to-real transformation:
\begin{equation}\label{channel_model}
    \mathbf{y} = \mathbf{H}\mathbf{x} + \mathbf{n},
\end{equation}
where 
\begin{equation*}
\mathbf{y} = \left[ {\begin{array}{*{20}{c}}
{\Re (\mathbf{\tilde y})}\\
{\Im (\mathbf{\tilde y})}
\end{array}} \right],~ 
\mathbf{x} = \left[ {\begin{array}{*{20}{c}}
{\Re (\mathbf{\tilde x})}\\
{\Im (\mathbf{\tilde x})}
\end{array}} \right],~ 
\mathbf{n} = \left[ {\begin{array}{*{20}{c}}
{\Re (\mathbf{\tilde n})}\\
{\Im (\mathbf{\tilde n})}
\end{array}} \right],
\end{equation*}

\begin{equation}
\mathbf{H} = \left[ {\begin{array}{*{20}{c}}
{\Re (\mathbf{\tilde H})}&{ - \Im (\mathbf{\tilde H})}\\
{\Im (\mathbf{\tilde H})}&{\Re (\mathbf{\tilde H})}
\end{array}} \right].
\end{equation}

Accordingly, the ML detector estimates the transmitted vector as $\mathbf{\hat x}(\mathbf{y})$ by solving the following problem:
\begin{equation}
    \mathbf{\hat x}(\mathbf{y}) = \mathop {\arg \min }\limits_{{\rm{\mathbf x}} \in \mathcal{S}^{2N_t} }  \left\| {{\rm{\mathbf y}} - {\rm{\mathbf{Hx}}}} \right\|_2^2, \label{ML_detection}
\end{equation}
where the minimization is over $\mathbf{x}\in\mathcal{S}^{2N_t}$, i.e., over all possible transmitted vectors. Notably, $\mathcal{S}=\Re (\mathcal{\tilde S})$, which is also equal to $\Im (\mathcal{\tilde S})$ in the complex-valued constellations. 

We consider employing the DetNet, a deep unfolding approach for the above ML detection~\cite{samuel2019learning}. It unfolds the PGD algorithm into cascaded blocks:
\begin{equation}
    \cM(\bH,\by)=\cM_L(\bH,\by,\cM_{L-1}(\bH,\by,\cdots\cM_1(\bH,\by)\cdots)),
\end{equation}
where $L$ is the number of required blocks. The operations within the $k$-th block, $\mathcal{M}_{k}$, are described as follows:
\begin{equation}
    {\mathbf{s}_k} = {\mathbf{x}_{k-1}} - {\alpha _{1k}}{\mathbf{H}^{T}}\mathbf{y} + {\alpha _{2k}}{\mathbf{H}^T}\mathbf{H}{\mathbf{x}_{k-1}}, \label{s_k}
\end{equation}
\begin{equation}
    {\mathbf{u}_k} \!=\! \left[\! {\begin{array}{*{20}{c}}
    {{\mathbf{s}_k}}\\
    {{\mathbf{a}_{k-1}}}
    \end{array}} \!\right] \!=\! \left[\! {\begin{array}{*{20}{c}}
    {{\mathbf{x}_{k-1}} \!-\! {\alpha _{1k}}{\mathbf{H}^T}\mathbf{y} \!+\! {\alpha _{2k}}{\mathbf{H}^T}\mathbf{H}{\mathbf{x}_{k-1}}}\\
    {{\mathbf{a}_{k-1}}}
    \end{array}} \!\right]\!, \label{u_k}
\end{equation}
\begin{equation}
    {\mathbf{z}_k} = \rho \left( {{\mathbf{W}_{1k}}{\mathbf{u}_k} + {\mathbf{b}_{1k}}} \right), \label{z_k}
\end{equation}
\begin{equation}
    {{\mathbf{x}_{k}}} = {\mathbf{W}_{2k}}{\mathbf{z}_k} + {\mathbf{b}_{2k}}, \label{x_k}
\end{equation}
\begin{equation}
    {{\mathbf{a}_{k}}} = {\mathbf{W}_{3k}}{\mathbf{z}_k} + {\mathbf{b}_{3k}}, \label{a_k}
\end{equation}
where $\rho (x) = \max (0, x)$ is the ReLU nonlinear activation function. Initially, we set $\mathbf{x}_0 = \mathbf{0}$ and $\mathbf{a}_0 = \mathbf{0}$. Notably, $\mathbf{a}_k$ serves as an auxiliary variable introduced to expand the input dimensionality. It acts as a hidden state that propagates through the network with $\mathbf{x}_k$, functioning as a bias-like term without direct physical meaning. The trainable parameters consist of the set $\Theta  = \left\{ {{\mathbf{W}_{1k}},{\mathbf{W}_{2k}},{\mathbf{W}_{3k}},{\mathbf{b}_{1k}},{\mathbf{b}_{2k}},{\mathbf{b}_{3k}},{\alpha _{1k}},{\alpha _{2k}}} \right\}_{k = 1}^L$. In this architecture, each block computes a linear combination of its inputs to yield an intermediate vector $\mathbf{s}_k$, which is then fed into a two-layer fully-connected neural network. Each iteration constitutes a computational block, and thus the block number $L$ represents the number of iterations.

\subsubsection{Circuit architecture}
The key circuit element for enabling IMC is the memristor, a non-volatile electronic device that stores information in its conductance state, which can be modified through programming operations and read in a non-destructive manner~\cite{chua2003memristor}. When assembled in a crossbar array structure, memristors enable highly parallel computing by performing data storage and certain computational tasks directly within memory. A representative crossbar array architecture employed for IMC implementation is the {one-transistor-one-resistor} (1T1R) array. As depicted in Fig. \ref{MVM}, each cross-point integrates a transistor and a memristor, where the {word-line} (WL) applies voltage to the transistor's gate to select the appropriate memristor. When a memristor is activated via its WL, the associated transistor is turned on, enabling current to flow through the memristor from its {bit-line} (BL) to {source-line} (SL) for read or programming operations, while unselected memristors remain isolated due to their transistors being in the off state.  

\begin{figure}[t!]
    \centering
    \includegraphics[width=0.8\columnwidth]{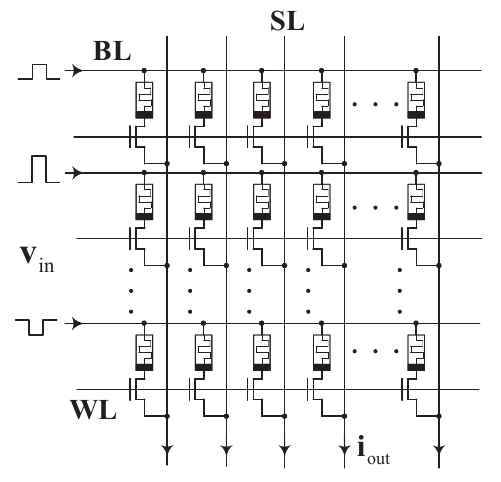}
    \caption{Circuit schematic of a 1T1R array. }
    \label{MVM}
\end{figure}

\subsubsection{MVM computation}
 
IMC leverages the parallelism of these memristor-based architectures to perform MVMs directly within the crossbar array in a single operation~\cite{verma2019memory}. To begin with, the real matrix is mapped into the conductances of memristors in a pair of crossbar arrays. The conductance matrix $\mathbf{G}$ is represented as the difference between two positive conductance matrices, i.e., $\bf G=G^+-G^-$. Next, the vector is mapped into an input voltage vector $\bf {v}$$_{\text{in}}$. By applying Ohm’s law and Kirchhoff’s law, the MVM result is obtained as the output current vector $\bf {i}$$_{\text{out}} = $$ \bf (G^+-G^-)$$\bf {v}$$_{\text{in}}$. Once the conductance matrix is configured, the aforementioned MVM can be executed in one step by applying one read pulse~\cite{zeng2023realizing}.
 
\subsubsection{Conductance programming}
Consider the conductance update process via a train of programming pulses.  Let $G_{\rm on}$ and $G_{\rm off}$ be the maximum and minimum conductance, respectively. Then, we adopt the behavioral model reported in~\cite{zeng2023realizing}, which characterizes the conductance change for each programming pulse as follows:
\begin{equation}
    \Delta G=\frac{G_{\mathrm{on}}-G_{\mathrm{off}}}{N_p}+n_w,
\end{equation}
where $N_p$ represents the number of programming pulses corresponding to programming the conductance from $G_{\rm off}$ to $G_{\rm on}$, and the per-cycle programming noise is modeled as $n_w \sim \mathcal{N} (0,\sigma _{{\Delta}g}^2)$. The C2C variation ${\sigma _{{\Delta}g}}$ quantifies the fluctuation in conductance change induced by each programming pulse, and is expressed as a percentage of the full conductance range~\cite{yu2018neuro}:
\begin{equation}
    {\sigma _{{\Delta}g}}={\gamma}(G_{\rm on}-G_{\rm off}),
\end{equation}
where the coefficient $\gamma$ typically falls within the interval (0, 0.05)~\cite{fu2022level}. This write-without-verification programming method enables high-speed programming while incurring large conductance variations. By contrast, other memristor non-idealities, such as state drift and limited endurance~\cite{verma2019memory}, exert relatively minor influence. Their effects can be mitigated through advanced device fabrication and material optimization. Hence, we only consider the conductance variations in the following analysis and simulations.

\begin{figure*}[t!]
    \centering
    \includegraphics[width=0.9\textwidth]{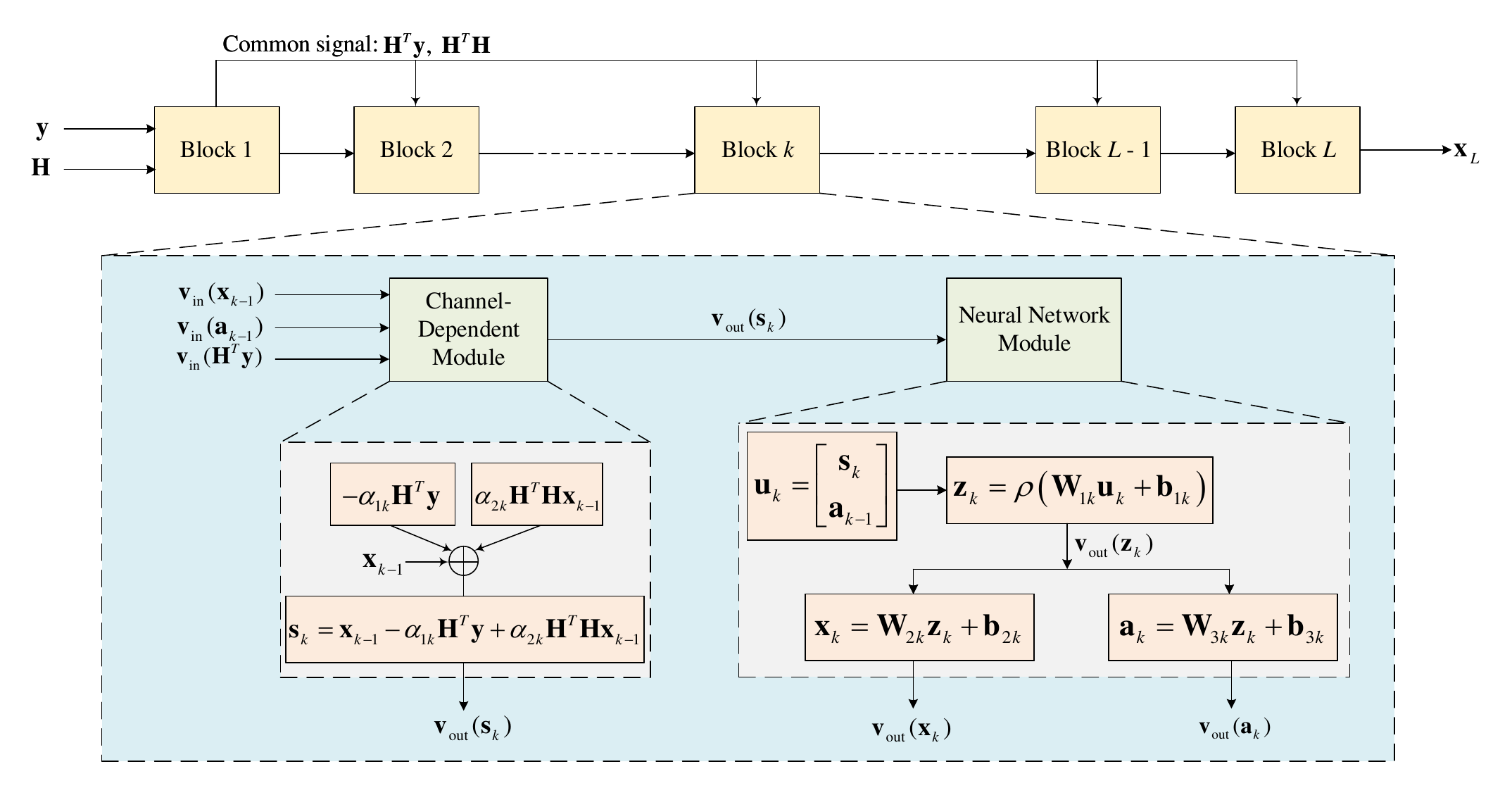} 
    \caption{System diagram of the proposed deep IM-MIMO detector.}
    \label{overview_circuit}
\end{figure*}

\subsection{Performance Metrics}
The performance of the proposed deep IM-MIMO detector is measured using the following three metrics.
\subsubsection{Detection accuracy} Bit error rate (BER) is defined as the ratio of the number of incorrectly received bits to the total number of transmitted bits, and is expressed as ${\rm{BER}} = \frac{{{N_{error}}}}{{{N_{total}}}} \times 100\%$, where $N_{error}$ and $N_{total}$ denote the number of error bits and the number of total bits, respectively. 
For tractability, we consider the output error between the clean and distorted outputs $\mathbf{e}_L=\mathcal{M}(\mathbf{H}+\Delta \mathbf{H},\mathbf{y}_0+\mathbf{n})-\mathcal{M}(\mathbf{H}, \mathbf{y}_0)$ in our analysis, where $\Delta \mathbf{H}$ is the distortion induced by programming noise and $\mathbf{y}_0=\mathbf{Hx}$ represents the received signal without channel noise. 
For codewords of fixed length, a higher error rate leads to a proportional increase in the expected number of bit flips, thereby resulting in a higher BER.
    
\subsubsection{Processing latency} The overall processing latency consists of three key components. One is the programming latency $T_{p}$ for programming the conductance matrices into the crossbar arrays, the computation latency $T_{c}$ for obtaining final outputs, and the latency from analog-to-digital converters (ADCs) and digital-to-analogue converters (DACs) at both ends of the circuit. The total processing latency is given by $T_{total} = T_{p}+T_{c}+T_{ADC}+T_{DAC}$.

\subsubsection{Hardware complexity} The hardware complexity is mainly determined by both the number and size of crossbar arrays, which influence both the number of required memristors and peripheral circuit components. The number of memristors is determined by the product of the number of rows and columns in each array. Meanwhile, crossbar arrays require peripheral circuits for effective operation, e.g., ADCs, DACs, {transimpedance amplifiers} (TIAs), and adders. The number of these components is also dictated by the number and the size of crossbar arrays.

\section{IMC Architecture for Deep IM-MIMO Detection}\label{IMC Architecture}
This section presents the architecture for the proposed deep IM-MIMO detector including the circuit design, and a novel noise-aware training scheme to enhance the detector's robustness against the programming noise.

\subsection{IMC Circuit Design}
As depicted in Fig. \ref{overview_circuit}, the designed deep IM-MIMO detector is composed of two modules according to their functionalities. The channel-dependent module undertakes computations involving the channel matrix, received signal vector, and estimated signal from the preceding block. The neural network module is dedicated to computations related to the neural network. The detailed circuit designs and underlying functionalities of these modules are elucidated as follows.

\subsubsection{Channel-Dependent Module}

\begin{figure*}[!t]
    \centering
    \includegraphics[width=0.65\textwidth]{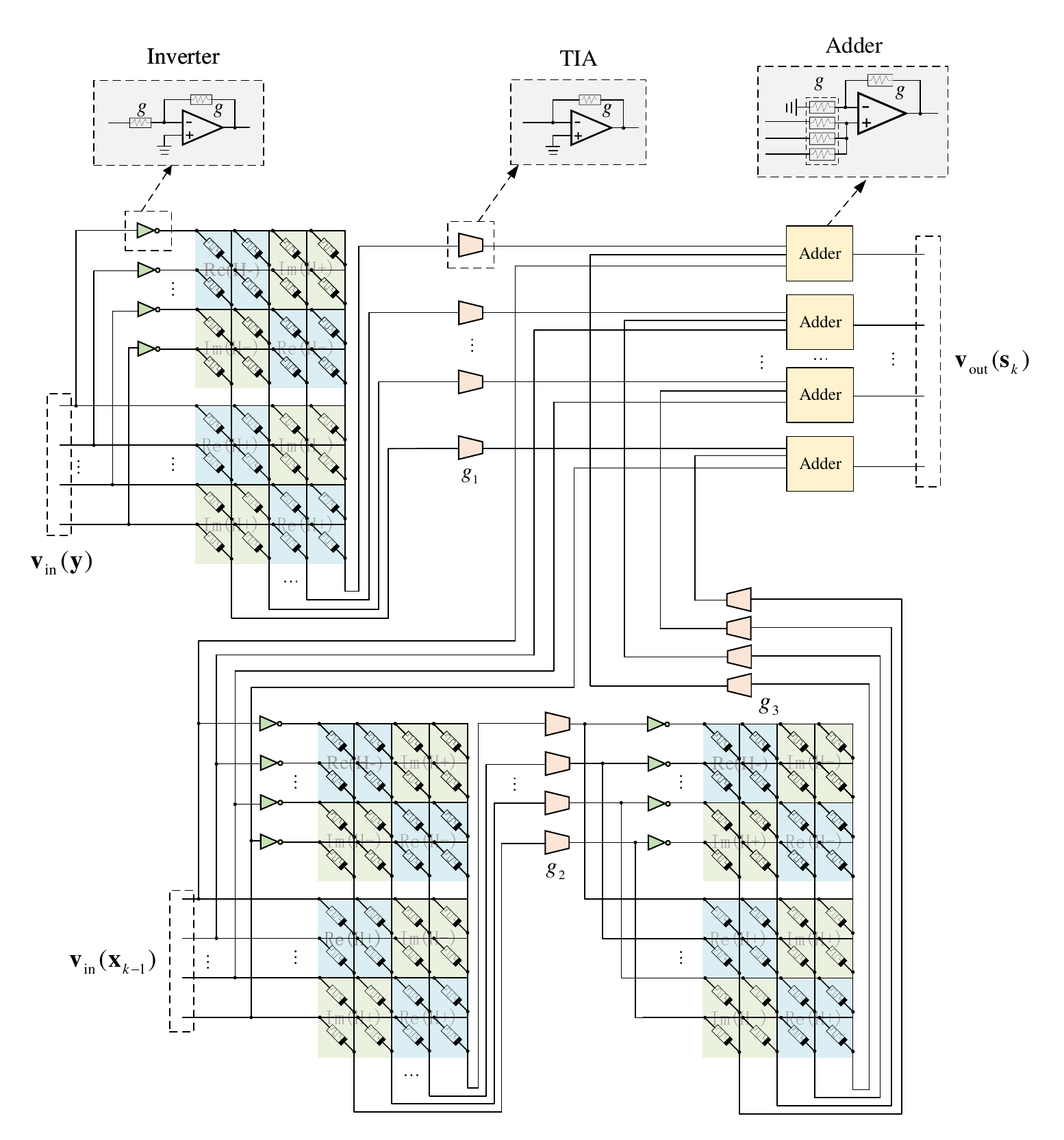} 
    \caption{The proposed sub-circuit for the channel-dependent module.}
    \label{CS_module}
\end{figure*}

The proposed channel-dependent module circuit is depicted in  Fig. \ref{CS_module}\footnote{To ensure clarity, a simplified crossbar array schematic is utilized as a concise yet representative visualization of the 1T1R architecture throughout the subsequent discussions.}. It comprises six $2N_r \times 2N_t$ crossbar arrays. In this configuration, TIAs are employed to cascade pairs of crossbar arrays while converting currents to voltages. This approach eliminates the need for ADCs, which are typically used in conventional architectures to obtain intermediate results (see e.g.,~\cite{jain2020rxnn}). As a result, our design avoids redundant A/D conversions. All TIAs share a similar configuration, differing in their respective feedback conductance $g$. The six crossbar arrays are organized into three groups, where each group comprises one array encoding the conductance matrix $\mathbf{H}^+$ and the other encoding $\mathbf{H}^-$. By utilizing inverters in this differential configuration, the circuit enables the mapping of matrices that include negative elements.

To compute ${-\alpha _{1k}}{\mathbf{H}^{T}}\mathbf{y}$ and ${\alpha _{2k}}{\mathbf{H}^T}\mathbf{H}{\mathbf{x}_{k-1}}$, we utilize the feedback resistors of TIAs to perform vector-scalar multiplication operations. Specifically, $\alpha _{1k}$  is mapped to $r_1=1/g_1$. Similarly, the computation of $\mathbf{v}_{\mathrm{out}}$$({\alpha _{2k}}{\mathbf{H}^T}\mathbf{H}{\mathbf{x}_{k-1}})$ requires passing through two sets of TIAs. Accordingly, $\alpha _{2k}$ is mapped to the product of $r_2=1/g_2$ and $r_3=1/g_3$. By applying Ohm’s law and Kirchhoff’s law, we derive the output voltages as $\bf{v}$$_{\mathrm{out}}$$({-\alpha _{1k}}{\mathbf{H}^{T}}\mathbf{y})=-$$r_1$$\mathbf{i}_{\mathrm{in}}(\mathbf{H}^{T}\mathbf{y})$ and $\mathbf{v}_{\mathrm{out}}$$({\alpha _{2k}}{\mathbf{H}^T}\mathbf{H}{\mathbf{x}_{k-1}})=$$r_2r_3$$\mathbf{i}_{\mathrm{in}}(\mathbf{H}^{T}\mathbf{H}\mathbf{x}_{k-1})$. Since both $\mathbf{y}$ and $\mathbf{H}$ remain constant across different blocks, the crossbar arrays corresponding to $\mathbf{H}$ can be reused.

Finally, $\mathbf v_{\mathrm{out}}({-\alpha _{1k}}{\mathbf{H}^{T}}\mathbf{y})$, $\mathbf v_{\mathrm{out}}({\alpha _{2k}}{\mathbf{H}^T}\mathbf{H}{\mathbf{x}_{k-1}})$, and $\mathbf{v}_{\mathrm{in}}(\mathbf{x}_{k-1})$ are summed using adders to compute $\mathbf v_{\mathrm{out}}(\mathbf{s}_k)$. The structure of the adder, designed to calculate the sum of three values, is depicted in Fig. \ref{CS_module} as well. All resistors within the adder have the same conductance $g_0$, ensuring that $\mathbf{v}_{\mathrm{out}}(\mathbf{s}_k)=\mathbf{v}_{\mathrm{out}}({-\alpha _{1k}}{\mathbf{H}^{T}}\mathbf{y})+\mathbf{v}_{\mathrm{out}}({\alpha_{2k}}{\mathbf{H}^T}\mathbf{H}{\mathbf{x}_{k-1}})+\mathbf{v}_{\mathrm{in}}(\mathbf{x}_{k-1})$.

\subsubsection{Neural Network Module}

\begin{figure*}[!t]
    \centering
    \includegraphics[width=0.9\textwidth]{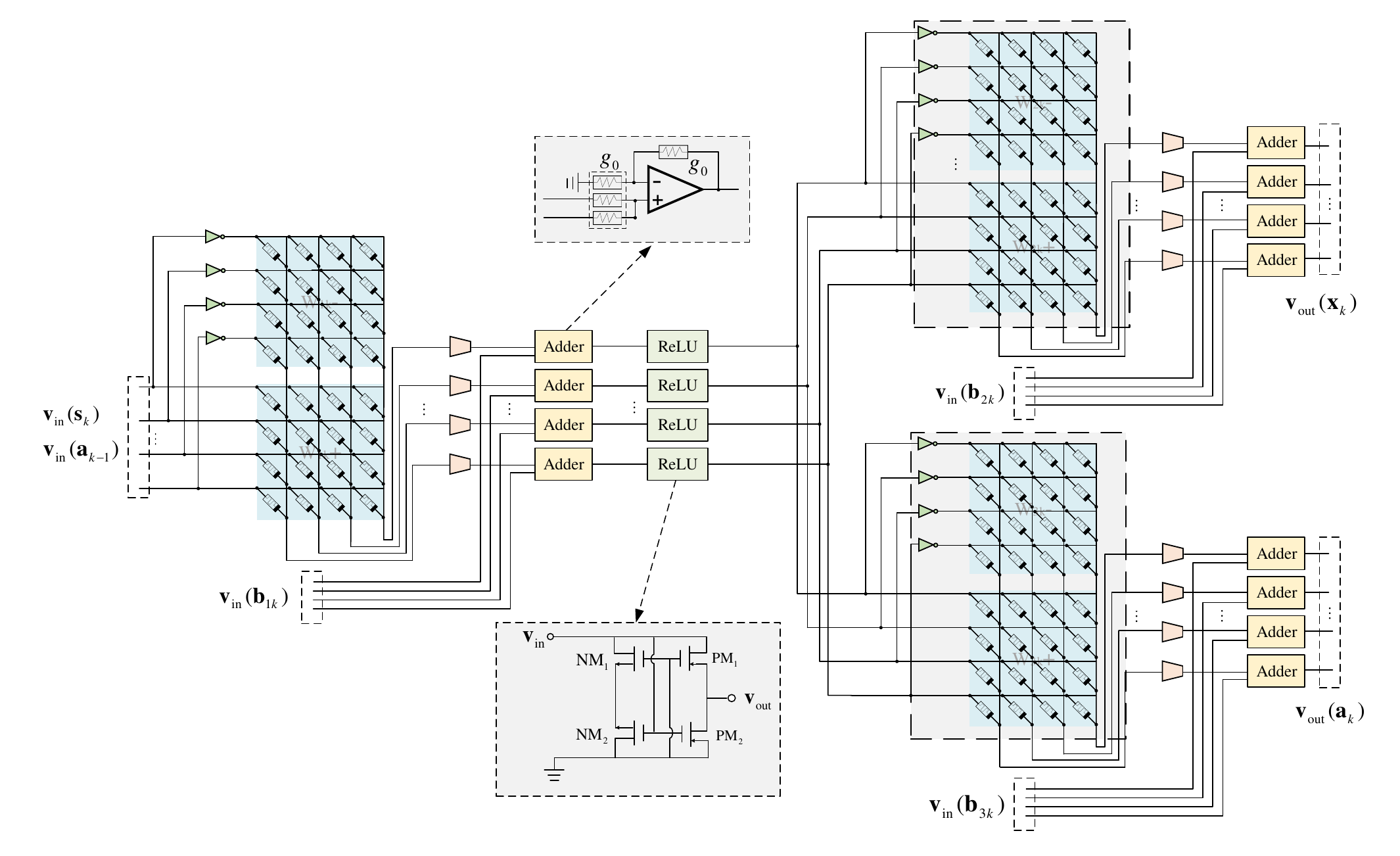} 
    \caption{The proposed sub-circuit for the neural network module.}
    \label{NN_module}
\end{figure*}

The proposed neural network module circuit is shown in Fig.~\ref{NN_module}. It is designed for implementing the affine transformation $\mathbf{W}\mathbf{x}+\mathbf{b}$ in \eqref{z_k}, \eqref{x_k} and \eqref{a_k}, and the nonlinear activation function ReLU in \eqref{z_k}. The concatenation operation of vectors $\mathbf{s}_k$ and $\mathbf{a}_{k-1}$ is achieved by directly increasing the input dimensionality and aligning them with the corresponding inputs. The two crossbar array groups responsible for (\ref{x_k}) and (\ref{a_k}) computations are designed to operate in parallel to enhance computational efficiency, as indicated by the dashed box annotations. The bias term $\mathbf{b}$ is seamlessly integrated into the computation of  $\mathbf{Wx}$ using dedicated adders. These adders are architecturally analogous to those described in the previous subsection, with the sole modification being their configuration to accommodate only two inputs.

For the ReLU nonlinear activation function referenced in (\ref{z_k}), its implementation in the analog domain is realized via a linear rectifier~\cite{huang2020analog}. Specifically, the linear rectifier employs a designed arrangement of two positive-channel metal-oxide-semiconductor (PMOS) transistors and two negative-channel metal-oxide-semiconductor (NMOS) transistors to realize the element-wise ReLU function. 

The final estimated signal vector, $\mathbf{x}_{L}$, output from the $L$-th block, $\mathbf{W}_{2L}\mathbf{z}_{2L}+\mathbf{b}_{2L}$, is converted to the digital domain using ADCs for subsequent demodulation. In this design, the proposed neural network module realizes fully-connected neural network operations, as defined by (\ref{z_k}), (\ref{x_k}) and (\ref{a_k}), based on the aforementioned principles.

\subsection{Noise-Aware Training Design} 
To improve detection accuracy of the proposed deep IM-MIMO detector under the influence of channel noise and crossbar array programming distortion, we propose a noise-aware training design. By leveraging the known statistical models of these non-idealities, this design mimics real-world noise scenarios by injecting statistically representative noise samples during the training process. The details are presented below.

Training is performed offline (i.e., not on the hardware platform) using a noise-aware training scheme that injects target noise during the training process.
The training dataset comprises batches of the transmitted signal vector $\mathbf{x}$, channel matrix $\mathbf{H}$, noise-free received signal vector $\mathbf{y}_0=\mathbf{Hx}$, channel noise $\mathbf{n}$, and channel matrix perturbation $\Delta \mathbf{H}$ induced by programming noise. 
During each training epoch, these components are independently generated according to their respective statistical distributions, with the distribution of $\Delta \mathbf{H}$ detailed in the subsequent section. The loss function is then computed and network parameters updated via backpropagation. To prevent convergence to local optima caused by initial unstable errors, we apply logarithmic weighting to each constituent block within the loss function. This strategy reduces the undue influence of large early-stage errors on optimization direction, leading to the following loss function definition:
\begin{equation}
    \mathcal{L}(\mathbf{x}; \mathbf{\hat x}(\mathbf{H}+\Delta \mathbf{H}, \mathbf{y}_0+\mathbf{n}; \Theta))={\sum\limits_{k = 1}^L {\ln (k)\| {{\mathbf{x}} - {\mathbf{\hat x}_k}} \|} ^2}. \label{loss}
\end{equation}

It is worth noting that the additional overhead of noise-aware training is negligible, as it comprises only random noise generation and element-wise addition to $\mathbf{H}$ during forward propagation. Upon completion of neural network training, the optimized parameters are programmed into the crossbar arrays for hardware deployment. 
Utilizing a programming-with-verification method \cite{zeng2023realizing} with sufficient configuration time enables high-precision weight mapping. During inference, these programmed weights remain static, while the system maintains adaptability to dynamic channel conditions through real-time updates of the time-varying channel matrix $\mathbf{H}$ across successive transmission slots based on channel state information (CSI).

As shown in Fig. \ref{robust}, the proposed training scheme exhibits enhanced robustness, maintaining BER performance stability for C2C variations up to a 2\% threshold. In contrast, the standard training approach without noise injection demonstrates noise robustness only up to 0.75\% C2C variation. This represents a 2.7$\times$ improvement in variation tolerance threshold compared to conventional training methods.

\begin{figure}[!t]
    \centering
    \includegraphics[width=0.8\columnwidth]{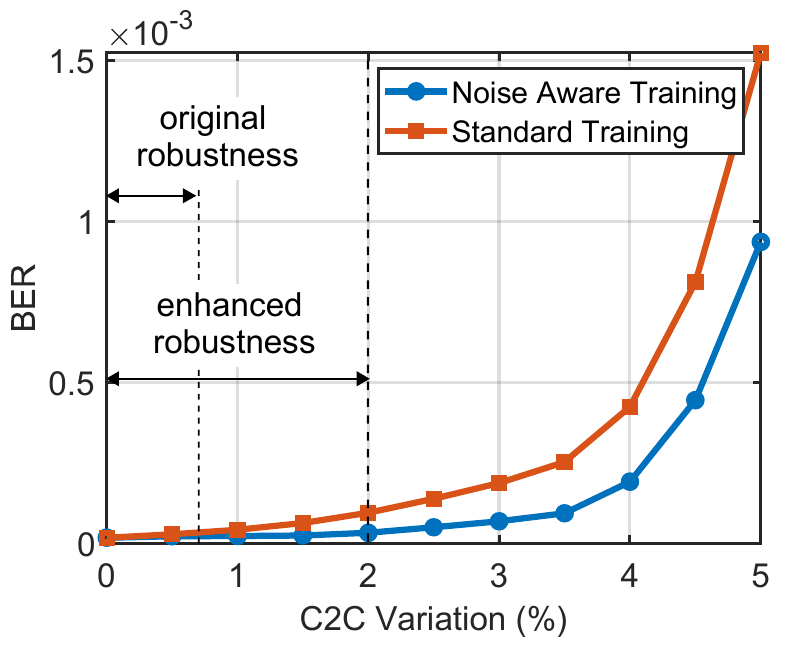} 
    \caption{Comparison of BER performance of two training schemes under different C2C variations.}
    \label{robust}
\end{figure}

\begin{remark}[Benefits from Neural Network Robustness]\rm
Neural networks inherently possess robustness, defined as the ability to maintain performance under noisy inputs within a certain tolerance. This property enables the system to accommodate both the channel matrix degraded by programming noise and received signal vectors corrupted by channel noise. Consequently, within the robustness threshold, additional accuracy compensation strategies (e.g., the programming-with-verification scheme for programming noise mitigation) are rendered unnecessary. 
\end{remark}

\section{Performance Analysis of Deep IM-MIMO Detection}\label{analysis}
In this section, we present a theoretical analysis to characterize the performance of the proposed ultra-low-latency deep IM-MIMO detector in terms of detection error, processing latency, and hardware complexity. 

\subsection{Detection Accuracy}
\subsubsection{Noise Characterization}\label{sec:perturbation_model}
In the context of the proposed deep IM-MIMO detector, we primarily focus on two types of noise, namely, AWGN originating from the wireless channel elaborated in Section \ref{Sec: PRELIMINARIES} and programming noise arising from the programming process of crossbar arrays.
In the channel-dependent module, the three-sigma rule is applied to ensure that each element $h_{ij}$ of $\mathbf{H}$ falls within the designated memristors' conductance range with a confidence level of 99.73\%. Accordingly, we define the mapping coefficient as $\mu  = ({G_{\rm on}} - {G_{\rm off}})/3$, which yields the following transformation: for $h_{ij}<0$, $G^+_{ij}=G_{\rm off}$ and $G^-_{ij}=G_{\rm off}+\mu |h_{ij}|$; for $h_{ij}>0$, $G^+_{ij}=G_{\rm off}+\mu h_{ij}$ and $G^-_{ij}=G_{\rm off}$. For each $h_{ij}$, only one memristor needs to be updated, i.e., $G^+_{ij}$ or $G^-_{ij}$, and the target conductance change is $\mu |h_{ij}|$. 

\begin{lemma} \label{distribution_of_h}\rm
    The programming errors in the channel matrix, $\{\Delta {h_{ij}}\}$, induced by cumulative C2C variations follow an independent but non-identical distribution:
\begin{equation}
    \Delta {h_{ij}}|{h_{ij}}\sim{\cal N}\left( {0,\sigma _{\Delta {h_{ij}}}^2} \right),\label{delta_h}
\end{equation}
where the variance depends on the target channel element as $\sigma _{\Delta {h_{ij}}}^2 = 3{\gamma ^2}{N_p}|{h_{ij}}|$.
\end{lemma}
\begin{proof}
    See Appendix A.
\end{proof}

\subsubsection{Upper Bound of Detection Error}
To establish the detection error propagation model, we transform $\mathbf{u}_k$ in (\ref{u_k}) into the following equation:
\begin{equation}
\begin{split}
\mathbf{u}_k &= \underbrace{\begin{bmatrix}
\mathbf{I}_{2N_r \times 2N_r} + \alpha_{2k} \mathbf{H}^T \mathbf{H} & \mathbf{0}_{2N_r \times 4N_r} \\
\mathbf{0}_{4N_r \times 2N_r} & \mathbf{I}_{4N_r \times 4N_r}
\end{bmatrix}}_{\mathbf{C}_1}
\begin{bmatrix}
\mathbf{x}_{k-1} \\
\mathbf{a}_{k-1}
\end{bmatrix} \\
&+ \underbrace{\begin{bmatrix}
- \alpha_{1k} \mathbf{H}^T \mathbf{y} \\
\mathbf{0}_{4N_r \times 1}
\end{bmatrix}}_{\mathbf{c}_2}.
\end{split}
\label{transformed}
\end{equation}
By setting ${{\bf{q}}_k} = \left[ {{{\bf{x}}_k^T},{{\bf{a}}_k^T}} \right]^T$, we have ${\mathbf{u}_k} = {\mathbf{C}_1}{\mathbf{q}_{k-1}} + {\mathbf{c}_2}$. By concatenating (\ref{x_k}) and (\ref{a_k}), we have ${{\mathbf{q}}_{k}} = \left[ {{{\mathbf{x}}_{k}^T},{{\mathbf{a}}_{k}^T}} \right]^T = {\mathbf{\hat W}_{2k}}{\mathbf{z}_{k}} + {\mathbf{\hat b}_{2k}}$, with ${\mathbf{\hat W}_{2k}}\stackrel{\Delta}{=}[\mathbf{W}_{2k}^T, \mathbf{W}_{3k}^T]^T$ and ${\mathbf{\hat b}_{2k}}\stackrel{\Delta}{=}[\mathbf{b}_{2k}^T, \mathbf{b}_{3k}^T]^T$. 
The concatenation of two parallel linear operations sharing a common input forms a unified linear transformation that is mathematically equivalent to the original structure. This functional equivalence preserves identical inference behavior in the original model. 
Consequently, the block-wise outputs $\{ {[{\mathbf{x}_k^T},{\mathbf{a}_k^T}]}^T\} _{k = 1}^L $ are generated through a two-layer fully-connected network.

Referring to (\ref{transformed}), the additive noise components $\Delta\mathbf{H}$ in $\mathbf{H}$ and $\mathbf{n}$ in $\mathbf{y}$ induce perturbations in $\mathbf{C}_1$ and $\mathbf{c}_2$, denoted respectively as
\begin{equation}
{ \Delta\mathbf{C}_1} = \left[ {\begin{array}{*{20}{c}}
{{\alpha _{2k}}\mathbf{({H^\textit{T}}} \Delta \mathbf{H} + \Delta \mathbf{{H^\textit{T}}H})}&\mathbf{0}\\
\mathbf{0}&\mathbf{0}
\end{array}} \right],\label{C1}
\end{equation}
\begin{equation}
    \Delta {\mathbf{c}_2} = \left[ {\begin{array}{*{20}{c}}
{ - {\alpha _{1k}}\mathbf{({H^\textit{T}}}\mathbf{n} + \Delta \mathbf{H^\textit{T}y_{\mathrm{0}})}}\\
\mathbf{0}
\end{array}} \right], \label{C2}
\end{equation}
where high-order small quantities $\Delta {\mathbf{H}^T}\Delta \mathbf{H}$ and $\Delta {\mathbf{H}^T}\mathbf{n}$ are negligible. 
The resulting error is composed of two parts. One contribution arises from the error propagation of the preceding block, while the other stems from the perturbations in the current block via $\Delta\mathbf{H}$ and $\mathbf{n}$. In our analysis, the weight and bias terms $\{\mathbf{W}_{1k}, \mathbf{\hat W}_{2k}, \mathbf{b}_{1k}, \mathbf{\hat b}_{2k}\}_{k=1}^L$ are assumed to be error-free, owing to offline fine-tuning of memristors before utilization. In a neural network, the spectral norm of weight matrices quantifies their capacity to scale and propagate errors through the layers, while the bias terms merely introduce an additive offset and do not influence the multiplicative error amplification. Therefore, we omit the explicit contribution of the bias terms for simplicity in the following analysis.
The output error of the $k$-th block, induced by the additional noise components $\Delta\mathbf{H}$ and $\mathbf{n}$, is defined as $\mathbb{E}\left\| {{\mathbf{e}_{k}}} \right\|_2=\mathbb{E}\| \mathcal{M}_k(\mathbf{H}+\Delta \mathbf{H},\mathbf{y}_0+\mathbf{n}, \mathbf{\hat x}_{k-1})-\mathcal{M}_k(\mathbf{H}, \mathbf{y}_0, \mathbf{\hat x}_{k-1}) \|_2$.

\begin{lemma}[Block-wise Error Propagation]\label{lemma2}\rm
    The error propagation of the $k$-th block can be upper bounded by
\begin{equation}
\begin{aligned}
    \mathbb{E}{\left\| {{\mathbf{e}_{k}}} \right\|_2} \leq { } &\mathbb{E}{\left\| {{\mathbf{\hat{W}}_{2k}}{\mathbf{D}_k}{\mathbf{W}_{1k}}}{(\mathbf{C}_1+\Delta\mathbf{C}_1)} \right\|_2} \cdot \mathbb{E}{\left\| {{\mathbf{e}_{k-1}}} \right\|_2} \\
    &+ \mathbb{E}{\left\| {{\mathbf{\hat{W}}_{2k}}{\mathbf{D}_k}{\mathbf{W}_{1k}}}{\Delta\mathbf{C}_1} \right\|_2} \cdot \mathbb{E}{\left\| {{\mathbf{q}_{k-1}}} \right\|_2} \\
    &+ \mathbb{E}{\left\| {{\mathbf{\hat{W}}_{2k}}{\mathbf{D}_k}{\mathbf{W}_{1k}}}\Delta {\mathbf{c}_2} \right\|_2}
\end{aligned} \label{e_k}
\end{equation}
where $\mathbf{q}_{k-1}$ is the noise-free output of the $(k-1)$-th block. $\mathbf{D}_k$ is the diagonal sign matrix with $(\mathbf{D}_{k})_{m,m}=1_{\{\mathbf{W}_{1k}\mathbf{u}_{k}\ge0\}}$, where $k\in\{1,2,\cdots,L\}$ and $m\in\{1,2,\cdots,S\}$.
\end{lemma}
\begin{proof}
    See Appendix B.
\end{proof}

The terms $\mathbf{C}_1$, $\Delta\mathbf{C}_1$, and $\Delta\mathbf{c}_2$ in (\ref{e_k}) are primarily derived from $\mathbf{H}$. We establish bounds on the spectral norms of $\mathbf{C}_1$ and $\Delta\mathbf{C}_1$ and the vector norm of $\Delta\mathbf{c}_2$ in the following lemmas.

\begin{lemma}[Scaling factor for the preceding block error]\label{lemma3}\rm
    The terms $\mathbb{E}\|\mathbf{C}_1\|_2$ and $\mathbb{E}\|\mathbf{C}_1+\Delta\mathbf{C}_1\|_2$, serving as scaling factors for the error propagated from the preceding block, are bounded by
    \begin{align}
    \mathbb{E}\|\mathbf{C}_1\|_2&\leq 2\varpi_2 \Phi^2\triangleq \varphi,\\
    \mathbb{E}\|\mathbf{C}_1+\Delta\mathbf{C}_1\|_2 &\leq \varpi_2 \Phi^2\left( \gamma\Phi \sqrt{6 \sqrt{ \frac{2}{\pi}} N_p} + 2 \right)\triangleq \tau,
    \end{align}
    where $\varpi_2=\max \alpha_{2k}$ and $\Phi=\sqrt{N_t} + \sqrt{N_r}$.
\end{lemma}
\begin{proof}
    See Appendix C.
\end{proof}

\begin{lemma}[Scaling factor in the current block]\label{lemma4} \rm
The term $\mathbb{E}\|\Delta \mathbf{c}_2\|_2$ serves as a scaling factor for the error arising from the current block, and it satisfies the following inequality 
\begin{equation}
\begin{split}
     &\mathbb{E}\|\Delta \mathbf{c}_2\|_2 \\
     &\le  2\varpi_1 \sqrt{N_t+N_r}\left( \sigma_n \sqrt{2 N_r} + \gamma  (N_t+N_r) \sqrt{3 \sqrt{\frac{2}{\pi}} N_p}\right)\\
     &\triangleq \xi,
\end{split}
\end{equation}
where $\varpi_1=\max \alpha_{1k}$.
\end{lemma}
\begin{proof}
    See Appendix D.
\end{proof}

Building upon the preceding lemmas, Theorem \ref{Theorem 1} establishes the detection error performance bound for the proposed deep IM-MIMO detector.

\begin{theorem}[Detection error]\label{Theorem 1}\rm
The detection error of the proposed deep IM-MIMO detection is defined as follows
\begin{equation}
\begin{split}
\mathbb{E}\|\mathbf{e}_L\|_2 &\leq \xi + \Gamma + \left[ \xi\cdot \frac{1-\tau L}{1-\tau} + \Gamma\cdot \left( L - \frac{1}{\varphi-1} \right) \right]\Omega \\
&\quad + O(\Omega^2),
\end{split}
\end{equation}
where 
\begin{equation}
    \Omega=\sqrt{\frac{4S}{3N_r}}e^{-S/8},
\end{equation}
\begin{equation}
    \Gamma=\frac{2\varpi_1 \varpi_2 \gamma\Phi^5(\varphi^L-\tau^L)}{(\varphi - 1)(\varphi-\tau)}\sqrt{6\sqrt{\frac{2}{\pi}}N_p}.
\end{equation}
\end{theorem}
\begin{proof}
    See Appendix E.
\end{proof}

Although the result in Theorem \ref{Theorem 1} seems complicated, we can observe some key insights from the asymptotic behavior with respect to different variables.

\begin{itemize}
    \item \textbf{MIMO scaling effect:} The detection error increases monotonically with the MIMO dimension $N\times N$, scaling asymptotically as $O(N^{3L/2})$. This result indicates that while massive MIMO configurations enhance spectral efficiency by leveraging an increased number of antennas to boost spatial multiplexing gain, they simultaneously amplify detection errors.
    \item \textbf{Channel noise effect:} Detection error grows with the noise standard deviation $\sigma_n$. Critically, as $\sigma_n \to 0$, only the term $4N\varpi_1\sigma_n$ in the error bound term $\xi$ vanishes. The persistent components demonstrate that channel noise reduction alone cannot eliminate the dominant error mechanisms.  
    \item \textbf{Memristor distortion effect:} The detection error increases with the C2C variation $\gamma$. The limit $\gamma \to 0$ yields a residual error bound:  $4\varpi_1\sigma_n N[ 1+ \left ( L+\frac{L-1}{8\varpi_2  N} \right )\sqrt{\frac{4S}{3N}} e^{-S/8}]$, proving that even idealized memristor crossbar arrays ($\gamma=0$) remain accuracy-constrained by the communication model.  
    \item \textbf{Crossbar architecture effect:} For the block number $L$, the $O(L\tau^L)$ scaling implies monotonic error reduction for $L > \lfloor -1/\ln \tau \rfloor$ in our design ($\tau < 1$). For the network size $S$, the asymptotic behavior $O(\sqrt{S}e^{-S/8})$ ensures error decreases exponentially with $S$. This theoretical relationship is further validated by the simulation results presented in Fig.~\ref{tradeoff}, which explicitly demonstrate that increasing $L$ and $S$ effectively reduces detection errors. However, these performance enhancements come at the cost of substantial hardware complexity overhead.
\end{itemize}
\begin{figure}[!t]
\centering
\subfloat[BER vs. block number $L$]{\includegraphics[width=0.8\columnwidth]{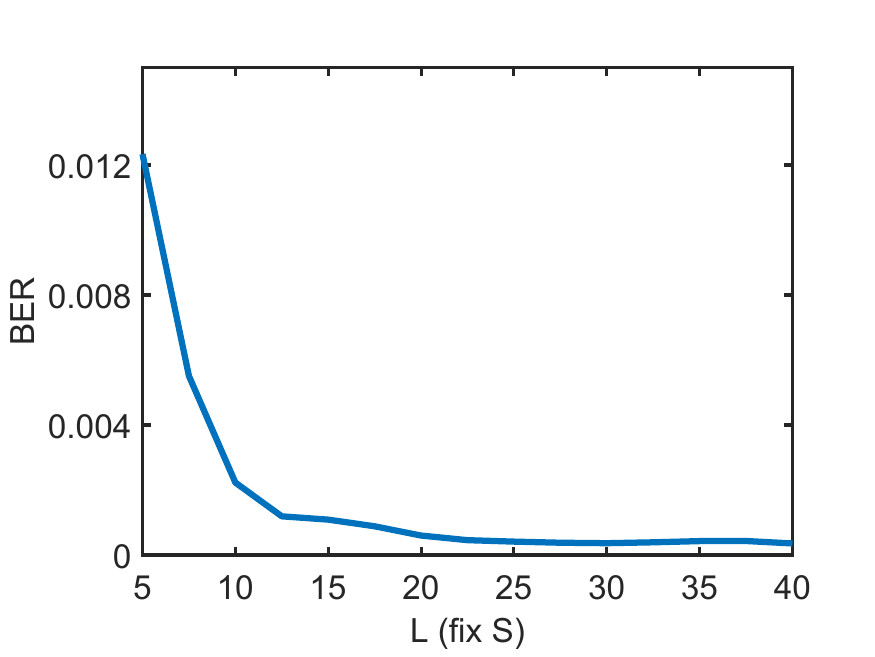}}
\label{block_num}
\hspace{-5mm}
\subfloat[BER vs. network size $S$]{\includegraphics[width=0.8\columnwidth]{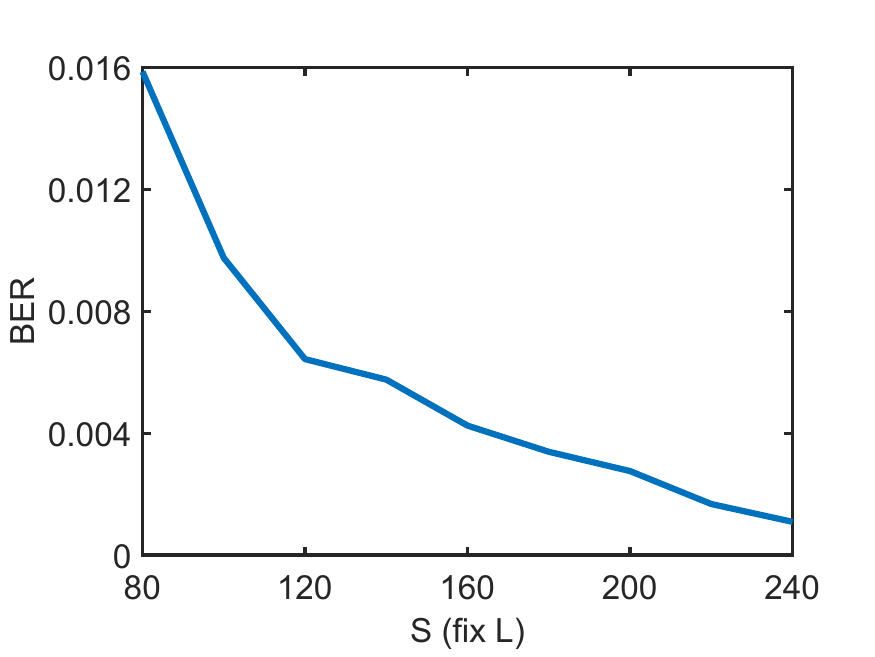}}
\label{network_size}
\caption{Comparison of the BER as a function of block number and network size in a $20\times 30$ Rayleigh fading channel with QPSK modulation. (a) BER vs. block number $L$ with fixed $S=240$. (b) BER vs. network size $S$ with fixed $L=15$.}
\label{tradeoff}
\end{figure}

\subsection{Processing Latency}

\subsubsection{Programming Latency} 
The channel matrix $\mathbf{H}$ is programmed into the crossbar array using a row-by-row programming-without-verification scheme, where WLs are sequentially activated to program individual rows. Following \cite{zeng2023realizing}, the programming latency for each row is determined by the maximum programming latency among all memristors within that row, which is characterized in the following lemma.

\begin{lemma} [\cite{zeng2023realizing}] \label{latency_one_row}\rm
    The expected latency for programming one row of a channel matrix $\mathbf{H}\in\mathbb{R}^{2N_r\times 2N_t}$ into the crossbar array is bounded by 
\begin{equation}
    {\mathbb{E} }[{T_{r}^i}]\le \frac{{\sqrt 2 {G_{\rm on}}{N_p}\Delta {t_w}}}{{3({G_{\rm on}} - {G_{\rm off}})}}\left( {\sqrt {\ln {N_t}}  + \frac{1}{{\sqrt \pi  \ln {N_t}}}} \right).
\end{equation}
\end{lemma}

The channel matrix $\mathbf{H}$ consists of $2N_r$ rows to be programmed, and the programming process is equivalent to two identical processes, where each is responsible for programming $N_r$ rows. Therefore, the programming latency for one matrix is defined as follows
\begin{equation}
    {T}_{m} = 2{N_r} \times {\mathbb{E} }[ {T_{r}^i} ], \quad i\in\{1,2,\cdots ,N_r\}.
\end{equation}
Given the parallel architecture for computing $\mathbf{H}^T\mathbf{y}$ and $\mathbf{H}^T\mathbf{H}\mathbf{x}_{k-1}$, the total programming latency is determined by the maximum latency of these two components, i.e., $\mathbf{H}^T\mathbf{H}\mathbf{x}_{k-1}$. Considering the sequential programming of the two crossbar array sets in $\mathbf{H}^T\mathbf{H}\mathbf{x}_{k-1}$, the upper bound of the total programming latency is given by
\begin{equation}
    T_{p}\le  \, \frac{4\sqrt{2} G_{\text{on}} N_p \Delta t_w N_r}{3(G_{\text{on}} - G_{\text{off}})} \left( \sqrt{\ln N_t} + \frac{1}{\sqrt{\pi}\ln N_t} \right).\label{LatencyWrite}
\end{equation}
As indicated in (\ref{LatencyWrite}), $T_{p}$ is upper-bounded by the dominant order term $ O(N_r\sqrt{\ln N_t})$. This signifies that the time complexity grows slightly faster than linear with the MIMO scale, representing a low-complexity characteristic.
Fig. \ref{Latency_trend} illustrates the programming latency across different scales of MIMO configurations. The simulation results align with the theoretical upper bound established in (\ref{LatencyWrite}), validating the effectiveness of the analysis. 

\begin{figure}[!t]
    \centering
    \includegraphics[width=0.8\columnwidth]{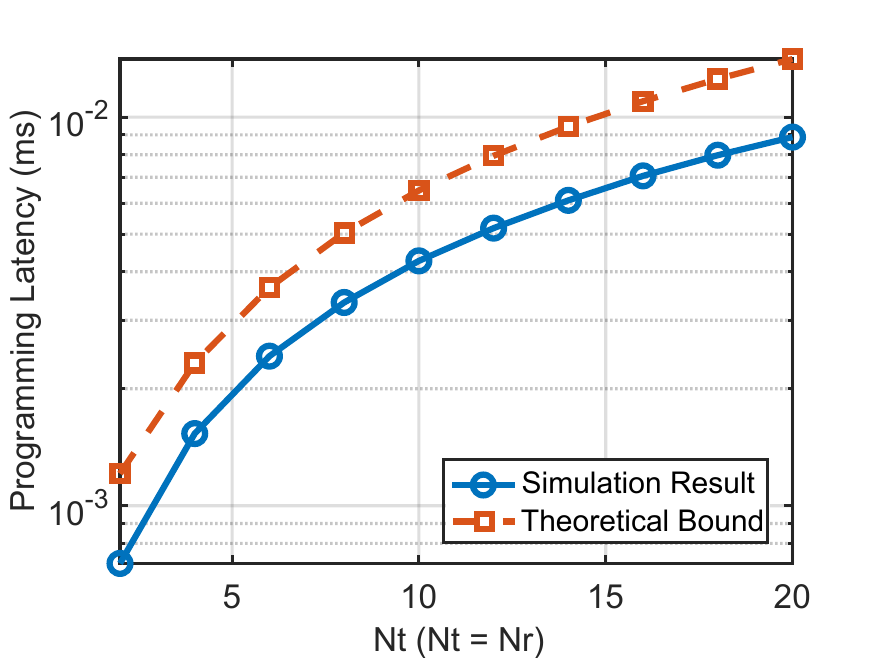} 
    \caption{Comparison of the theoretical bound and simulation results in terms of programming latency under different MIMO scales.}
    \label{Latency_trend}
\end{figure}

\subsubsection{Computation Latency} Computation latency is primarily determined by the circuit's stabilization time, which arises from cumulative delays introduced by circuit components, e.g., memristor-based crossbar arrays, adders, and ReLU circuits. The delays of these components are typically on the order of nanoseconds, a magnitude comparable to the duration of a programming pulse. The computation latency is approximately expressed as 
\begin{equation}
    T_{c}=L(t_{array}+t_{adder}+t_{ReLU}).
\end{equation}

\subsubsection{Total Processing Latency} 
$T_{ADC}$ and $T_{DAC}$ are independent of the circuit scale and fixed at both sides. The computation latency for each circuit block is typically on the nanosecond scale~\cite{zhan2022perovskite}, whereas the programming latency is generally on the order of microseconds. For small block counts $L$, the cumulative computation latency remains negligible, making programming latency dominant in the overall processing latency. However, as the number of blocks increases and the cumulative computation latency reaches the microsecond level, both computation and programming latencies become comparable, leading to a combined scaling of $O(N_r\sqrt{\ln{N_t}})+O(L)$.
\begin{remark}\rm
    To enhance detection accuracy, two major strategies can be considered:
    i) Improved crossbar array programming precision via programming-with-verification. This scheme mitigates the programming noise in memristors at the expense of increased programming latency, whose scaling law increases from $O(N_r\sqrt{\ln N_t})$ to $O(N_r\ln N_t)$~\cite{zeng2023realizing}.
    ii) Network architecture refinement through increased depth $L$ or width $S$. While larger $L$ improves accuracy, it increases computation latency at a scale of $O(L)$. Moreover, larger $S$ enhances accuracy but extends circuit stabilization time, proportionally increasing computation latency.
\end{remark}

\begin{table*}[!t]
\caption{Hardware Cost\label{hardware}}
\centering
\begin{tabular}{|c|c|c|c|}
\hline
Components & Number & Power/Component & Area/Component \\
\hline
Memristor~\cite{luo2022high} & $24N_rN_t +2L[ S(2N_r + 4N_t)+2N_tS+ 4SN_t]$ & 0.56 $\mu$W & 1960 $\mathrm{nm}^2$ \\
\hline
TIAs & $2N_r+4N_t+L(8N_t+2N_r+S)$ & 357 $\mu$W & 663 $\mu\mathrm{m}^2$ \\
\hline
Adder & $L(8N_t+S)$ & 459.2 $\mu$W & 680 $\mu\mathrm{m}^2$ \\
\hline
ReLU circuit~\cite{huang2020analog} & $SL$ & 18.4 nW & 4.88 $\mu\mathrm{m}^2$ \\
\hline
ADC~\cite{choo201627} & 1 & 1.97 mW & 906.3 $\mu\mathrm{m}^2$ \\
\hline
DAC~\cite{huang2017high} & 2$N_r$ & 2.34 mW & 0.19 $\mathrm{mm}^2$ \\
\hline
\end{tabular}
\end{table*}

\subsection{Hardware Complexity}
The hardware complexity analysis focuses on memristor crossbar array dimensions and peripheral circuit requirements. 
The overall hardware cost is summarized in Table \ref{hardware}.
The hardware costs of all circuit components are normalized to the 50 nm technology node used for the memristor array, following Dennard scaling laws~\cite{dennard2003design}.

\subsubsection{Memristor crossbar array} 
In the channel-dependent module, there are six crossbar arrays of size \(2N_r \times 2N_t\), which can be reused across subsequent blocks. 
In the neural network module, there are three groups of crossbar arrays with sizes \(S\times (2N_r + a_{size})\), \(2N_t \times S\), and \(a_{size} \times S\), respectively. 
Here, $a_{size}$ denotes the length of the vector $\mathbf{a}_k$. Accordingly, the total number of memristors is given by $24N_rN_t +2L[ S(2N_r + a_{size})+2N_tS+ Sa_{size}]$. 

\subsubsection{Peripheral circuits} 
The cost of peripheral circuits also scales with the size and number of crossbar arrays. 
For the channel-dependent module, the six crossbar arrays are paired with $4N_t+2N_r$ TIAs, and $2N_t$ adders. Within the neural network module, there are $S+2N_t+a_{size}$ TIAs, $S+2N_t+a_{size}$ adders, and $S$ ReLU circuits. 
The amplifier \cite{feinberg2021analog} with a 500 MHz gain-bandwidth product (GBW) and 60 dB gain is used in adders and TIAs.
There is one ADC~\cite{choo201627} at the end of the circuit for read and $2N_r$ DACs~\cite{huang2017high} to provide input voltages. The ADC in the read process is shared across all the columns of the array.

As revealed in Theorem 1, increasing $L$ and/or $S$ improves the detection accuracy. However, the hardware cost, including power and area, increases with the number of components (see Table~I). Thus, there exists a tradeoff between accuracy performance and hardware complexity.

\section{Simulation Results}\label{results}
In this section, we evaluate the performance of the proposed deep IM-MIMO detector through extensive simulations. We begin with providing comparisons of time complexity with classical MIMO detectors. Following that, we present comparisons with SOTA ASIC implementations and commercial digital processors, followed by an accuracy comparison against conventional detectors.

\subsection{Simulation Settings}
The simulations are conducted on a Python-LTSpice co-simulation framework, where the PyLTSpice toolkit facilitates bidirectional data exchange \cite{brum2024pyltspice}. For different memristor configurations, we adopt a behavioral model that captures the dynamics of memristors through numerical approximations derived from their intrinsic material properties and operational characteristics. The C2C variation coefficient $\gamma$ considered in the simulation is set to $0 \sim 4\%$. The training process involves 25,000 epochs, utilizing the Adam optimizer during backpropagation with a training batch size of 500. The starting learning rate is set to 0.0008. Training is performed within a limited signal-to-noise ratio (SNR) range of 8 dB to 13 dB. During inference, the considered SNR range spans 2 dB to 14 dB. Performance evaluation is conducted under Rayleigh fading channels, considering BPSK, QPSK, and 16QAM for modulation schemes.

\begin{figure}[!t]
    \centering
    \includegraphics[width=0.8\columnwidth]{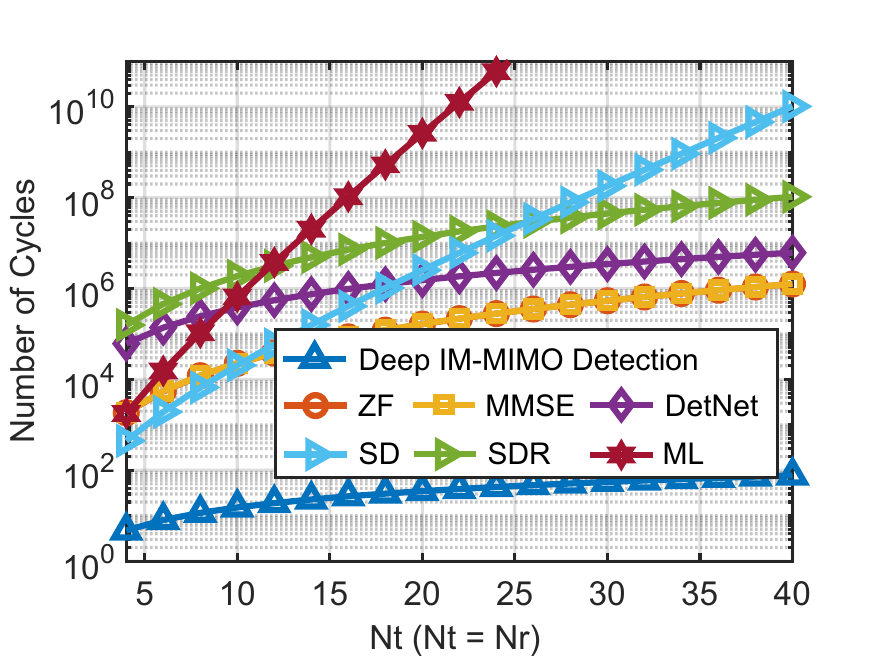} 
    \caption{Time complexity comparison of different detectors as a function of MIMO scales.}
    \label{complexity_figure}
\end{figure}

\subsection{Time Complexity Evaluation}

\begin{table*}[!t]
\centering
\begin{threeparttable}[b]
\caption{Time Complexity Comparisons of Different MIMO Detectors. \label{tab:table1}}
\begin{tabular}{| >{\centering\arraybackslash}m{5cm} | >{\centering\arraybackslash}m{7cm} |}
\hline
MIMO Detector & Time Complexity \\
\hline
ZF / MMSE & $O(N^3)$  \\
\hline
SDR & $O(N^3N_{\rm iter})$ \\
\hline
SD & $O(M^{\beta N})$, $\beta \in (0,1]$ \\
\hline
ML & $O(M^{N})$\tnote{(a)}\\
\hline
DetNet & $O(N^2L)$  \\
\hline
Proposed Deep IM-MIMO Detector & $O(N\sqrt {\ln {N}})+O(L)$ \\
\hline
\end{tabular}

\begin{minipage}{\textwidth}
\centering
\footnotesize
\begin{tabular}{@{}p{0.71\textwidth}}
\textsuperscript{(a)} The time complexity of the ML detector is much higher than that of other detectors, making it impractical under our simulation settings. Thus, we select SD as a near-ML performance benchmark that remains computationally tractable for subsequent comparisons. 
\end{tabular}
\end{minipage}

\end{threeparttable}
\end{table*}

For clarity, we consider a symmetric MIMO configuration where $N_r=N_t=N$. Table \ref{tab:table1} compares the time complexity of various MIMO detectors, which are categorized as follows.
\begin{enumerate}
    \item \textit{Linear detectors} (\textbf{MMSE} and \textbf{ZF}): Their complexity is constrained by matrix inversion, typically exhibiting a complexity of $O(N^3)$.

    \item \textit{Sphere decoder} (\textbf{SD}): The sphere decoder implementation follows \cite{studer2006soft}. Its complexity is primarily governed by search strategies, with exponential scaling $O(M^{N})$ in the worst case, which is equivalent to that of the \textbf{ML} detector. $M$ is the modulation order.

    \item \textit{Semi-definite relaxation} (\textbf{SDR}) detector: It is optimized from ML based on relaxed semi-definite programming algorithm \cite{ma2004semidefinite}. The SDR detector is a suboptimal detection method with a time complexity of $O(N^3N_{\rm iter})$, where $N_{\rm iter}$ denotes the preset number of optimization iterations. 

    \item \textbf{DetNet}: Representing the digital computing counterpart of the proposed deep IM-MIMO detector, DetNet exhibits a time complexity linearly proportional to the number of blocks $L$, specifically $O(N^2L)$.
\end{enumerate}

The linear dependence of DetNet's time complexity on $L$ can be alleviated by employing the one-step analog IMC paradigm. As illustrated in Fig. \ref{complexity_figure}, the time complexity of various detectors is evaluated by quantifying the number of cycles as a function of MIMO scale. In the digital domain, processing time is determined by clock cycles\footnote{We consider the scalar processor to calculate the number of cycles in the digital domain.}, while in the analog domain, each cycle corresponds to a programming pulse. Although the physical implementations of analog and digital domains differ significantly, both abstract time complexity through recurring cycle units. As demonstrated in Fig. \ref{complexity_figure}, the proposed deep IM-MIMO detector achieves the lowest time complexity among the considered detectors. 

\subsection{Processing Efficiency Evaluation}
\subsubsection{Comparison with ASIC Designs}

\begin{table*}[!t]
\caption{Parameters of the Target Memristor's Behavioral Model. \label{device}}
\centering
\begin{tabular}{|c|c|c|c|c|c|c|}
\hline
 Memristor&Pulse Width & C2C Variation & $G_{\rm off}$ & $G_{\rm on}$ & State Number \\
\hline
RRAM \cite{zeng2023realizing} & 10 ns & 4.41\% (P) 5.44\% (D) & 79.93 $\mu S$ & 230.99 $\mu S$ & 256 \\
\hline
FeFET \cite{jerry2017ferroelectric} & 75 ns & 0.5\% & 0.04 $\mu S$ & 1.79 $\mu S$ & 32 \\
\hline
FTJ \cite{luo2022high}&630 ps  & 3.65\% & 1 $\mu S$ & 27.5  $\mu S$& 150\\
\hline
\end{tabular}
\end{table*}

Fig. \ref{Latency_comparison} presents the programming latency comparison of different memristors:
resistive random-access memory (RRAM)~\cite{zeng2023realizing}, ferroelectric field-effect transistor (FeFET)~\cite{jerry2017ferroelectric}, and ferroelectric tunnel junction (FTJ)~\cite{luo2022high}. The relevant parameters are shown in the Table \ref{device}. Notably, the memristor in~\cite{zeng2023realizing} has different C2C variations during potentiation (P) and depression (D), respectively. It is observed that the memristor in~\cite{luo2022high} achieves a better balance between programming latency and detection accuracy and has excellent endurance (over $10^9$ cycles). Therefore, we adopt the memristor in~\cite{luo2022high} as the target memristor device unless specified otherwise.

\begin{figure}[!t]
\centering
\subfloat[BER vs. SNR]{\includegraphics[width=0.8\columnwidth]{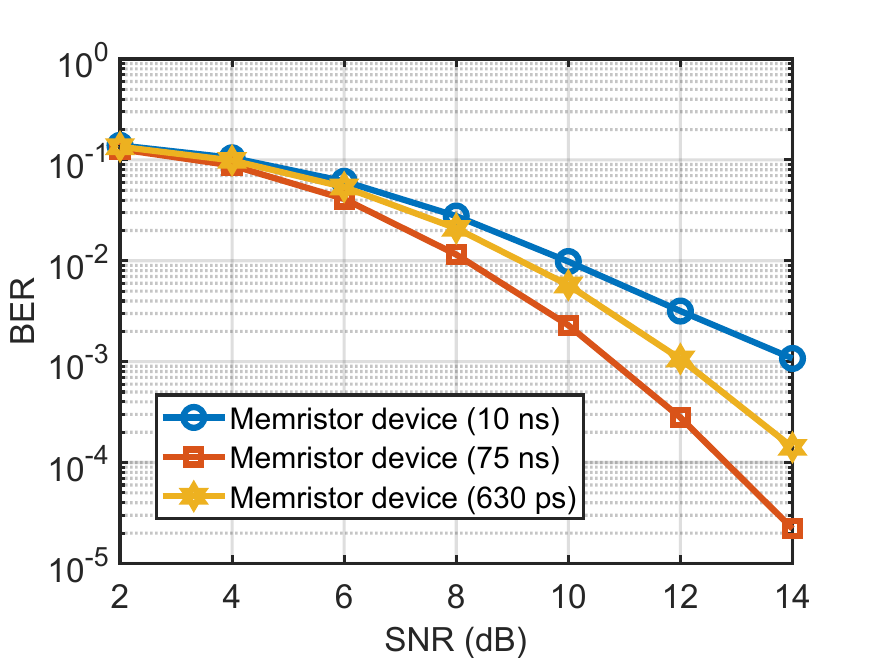}}
\label{ber_snr}
\hspace{-5mm}
\subfloat[Latency vs. MIMO scales]{\includegraphics[width=0.8\columnwidth]{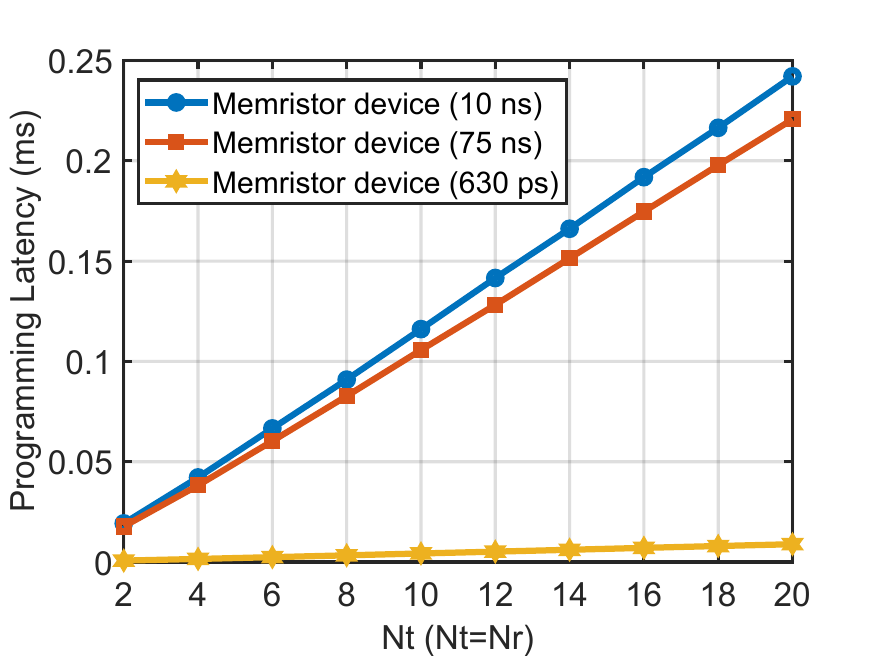}}
\label{latency_size}
\caption{Performance comparisons. (a) BER vs. SNR, $20\times 30$ Rayleigh fading channel with QPSK modulation, $L=30$, $S=480$. (b) Programming latency vs. MIMO scales.}
\label{Latency_comparison}
\end{figure}

\begin{table*}[!t]
    \centering
    \caption{Parameters for MIMO System Simulation}
    \begin{tabular}{|c|c|c|c|c|}
    \hline
       Parameter  & \# Sub-carriers & \# Symbols/frame & Modulation order & \# Tx antennas \\
       \hline
       Notation & $N_c$ & $M$ & $b$ & $N_t$ \\
       \hline
       Value  & 1024 & $160\times 14$ & 4 & 20 \\
       \hline
    \end{tabular}
    \label{parameters}
\end{table*}

To ensure a fair comparison, one frame of a 16-QAM OFDM system with 1,024 subcarriers is considered. 
The simulation parameters are summarized in Table~\ref{parameters}.
The discrete Fourier transform (DFT) module is implemented following the design described in~\cite{zeng2023realizing}. 
The number of floating-point operations (FLOPs) required per symbol is given by $L(4N_t^2+24N_tS+8N_t+2S)$. 
For system parameters $N_t = 20, N_r = 30, L=30, S = 480$, the total computational complexity amounts to 0.535 tera operations (TOPs), while the total number of transmitted bits is $N_c \times M \times b \times N_t = 183500800\,{\rm bits}\approx 0.184$ Gb. 
The estimated energy consumption includes programming energy, Joule heating of memristors, and energy consumption of other components. Joule heating of each device is estimated as 40 fJ.
The total area is estimated to be 38.56 $\rm mm^2$, and total energy consumption is about 0.029 J.
Thus, comparisons of energy efficiency, area efficiency, and throughput between the deep IM-MIMO detector with the SOTA MIMO detectors are summarized in Table~\ref{asic_11_normalized}. 
Notably, the performance results reported for reference designs are acquired from post-silicon characterization measurements, while those of the proposed design in this work are generated from pre-silicon analytical and simulation-based estimation.
Fig. \ref{ASIC_accuracy_compare} presents BER comparisons between DetNet and four representative ASIC detectors over a Rician fading channel (Rician factor $K=3$). 
It can be observed that DetNet consistently outperforms baselines, especially in the high SNR regime. 
Compared with the SOTA ASIC detectors, the deep IM-MIMO detector achieves a throughput improvement ranging from 3.6$\times$ to 159.3$\times$. 
In terms of energy efficiency, it demonstrates 8.1$\times$, 6.8$\times$, and 4.6$\times$ higher efficiency than the designs in~\cite{jeon2019354, li2025deep, prabhu20173}, respectively. Moreover, the proposed design outperforms \cite{jeon2019354} and \cite{lee202540} by 3.9$\times$ and 1.8$\times$ in area efficiency.

\begin{figure}[!t]
        \centering
        \includegraphics[width=0.8\columnwidth]{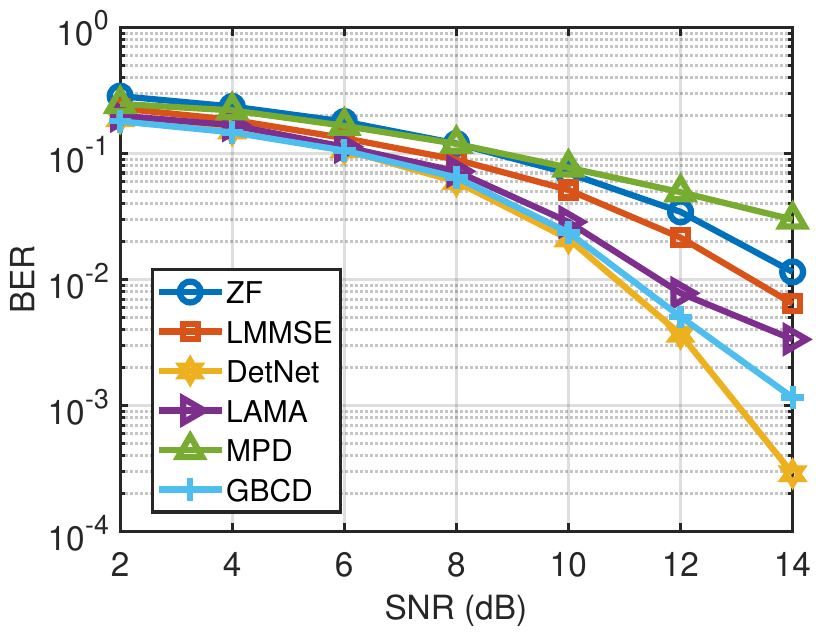} 
        \caption{BER comparisons with SOTA ASIC detectors in a $(N_t, N_r)=(20,30)$ system.}\vspace{4mm}
        \label{ASIC_accuracy_compare}
\end{figure}

\begin{table*}[!t]
        \centering
        \resizebox{\linewidth}{!}{
        \begin{threeparttable}[b]
        \caption{Comparison of SOTA ASIC MIMO Detectors.}
        \label{asic_11_normalized}
        \begin{tabular}{| >{\centering\arraybackslash}m{1.3cm} | >{\centering\arraybackslash}m{1.5cm} | >{\centering\arraybackslash}m{1.4cm} | >{\centering\arraybackslash}m{2.2 cm} | >{\centering\arraybackslash}m{2.0cm} | >{\centering\arraybackslash}m{2.2cm} | >{\centering\arraybackslash}m{2.2cm} | >{\centering\arraybackslash}m{2.2cm} |}
        \hline
         & Technology [nm] & MIMO Size & Modulation & Algorithm & Throughput\tnote{(g)}\quad\quad[Gb/s] & Energy Efficiency\tnote{(g)} [pJ/b] & Area Efficiency\tnote{(g)} [Gbps/$\textrm{mm}^2$] \\
         \hline
        Jeon~\cite{jeon2019354}   & 28 & $256 \times 32$ & 256 QAM                & LAMA\tnote{(a)}   & 0.3   & 1273.52 & 0.318 \\
    \hline
    Prabhu~\cite{prabhu20173} & 28 & $128 \times 8$  & 256 QAM                & ZF/MMSE           & 1.1   & 717.49  & -- \\
    \hline
    Li~\cite{li2025deep}     & 22 & $128 \times 16$ & 4 to 256 QAM           & GBCD\tnote{(b)}   & 13.3  & 1073.57 & 0.756 \\
    \hline
    Tang~\cite{tang20210}   & 40 & $128 \times 32$ & 256 QAM                & MPD\tnote{(c)}    & 2.6   & 116.89  & 3.249 \\
    \hline
    Lee~\cite{lee202540}    & 40 & $256 \times 32$ & 256 QAM                & MPD\tnote{(c)}               & 6.0   & 30.03   & 0.676 \\
    \hline
    Tang~\cite{tang20181}   & 28 & $128\times 16$  & QPSK to 256 QAM        & EPD\tnote{(d)}    & 3.4   & 41.16  & 1.531 \\
    \hline
    Yun~\cite{yun20245}   & 28 & $128\times 32$  & 4 to 256 QAM           & Deep learning     & 5.4   & \textbf{23.44}   & 3.136 \\
    \hline
    Wen~\cite{wen20201}   & 40 & $256\times 32$  & 4 to 256 QAM           & DCD\tnote{(e)}    & 1.8   & 26.40   & \textbf{4.480} \\
    \hline
    This Work & 50 & $20 \times 30$ & BPSK/QPSK, 4 to 16 QAM & PGD\tnote{(f)} & \bf{47.8} & 157.51 & 1.240 \\
        \hline
        \end{tabular}
        
        \begin{minipage}{\textwidth}
        \centering
        \footnotesize
        \vspace{2pt}
        \begin{tabular}{@{}p{0.5\textwidth}@{\hspace{0.001\textwidth}}p{0.5\textwidth}@{}}
        \textsuperscript{(a)} Large-MIMO approximate message passing. & 
        \textsuperscript{(b)} Gram-domain block coordinate descent. \\
        \textsuperscript{(c)} Message-passing detection. &
        \textsuperscript{(d)} Expectation-passing detection.  \\
        \textsuperscript{(e)} Dichotomous coordinate descent.& \textsuperscript{(f)} Projected gradient descent. \\
        \multicolumn{2}{@{}p{\linewidth}@{}}{\textsuperscript{(g)} Normalized to 50 nm technology node and $N_r=30$ receive antennas. 
        First, technology nodes are normalized to 50 nm according to Dennard scaling laws~\cite{dennard1999design}, where both area and power consumption are proportional to $\sim 1/k^2$ ($k$ is the ratio of technology nodes). 
        Then, receive antenna count $N_r$ is normalized to 30: throughput scales linearly with $N_r$, while energy and area scale with $N_r^2$.} \\
        \end{tabular}
        \end{minipage}
        \end{threeparttable}}
        \vspace{10pt}
\end{table*}

\subsubsection{Comparison with Digital Processors}
To provide a comprehensive performance benchmark, three representative digital processing platforms are selected. The Texas Instruments TMS320C6678 multi-core digital signal processor (DSP), as documented in~\cite{fryza2014power}, serves as a high-performance DSP solution. The Xilinx Kintex-7 480T field-programmable gate array (FPGA) is included due to its balanced trade-off among performance, power consumption, and cost, which contains only programmable logic (PL), and the design is described in hardware description language (HDL)~\cite{bertendsp2016gpu}. In addition, the NVIDIA Jetson Orin NX 16GB, leveraging its AI acceleration capabilities in small-form-factor scenarios, is widely adopted in practical AI applications~\cite{jetson_orin_nx_datasheet}.
An extended comparative evaluation summarized in Table~\ref{tab:counterpart} benchmarks the proposed deep IM-MIMO detector against these three digital processors in terms of energy efficiency and throughput, under the same target detection accuracy. Experimental results indicate that the proposed deep IM-MIMO detector achieves significant performance gains: offering 77.4$\times$, 74.1$\times$, and $10^{3}\times$ higher throughput than TMS320C6678, Kintex-7 480T, and NVIDIA Jetson Orin NX 16GB, respectively. In terms of energy efficiency, the proposed design surpasses TMS320C6678, Kintex-7 480T, and NVIDIA Jetson Orin NX 16 GB by factors of $1.1\times10^{4}$, 76.18 and 245.38, respectively.

\begin{table*}[!t]
\caption{Comparison with Digital Processors.\label{tab:counterpart}}
\centering
\begin{tabular}{| >{\centering\arraybackslash}m{3cm} | >{\centering\arraybackslash}m{2.3cm} | >{\centering\arraybackslash}m{2.3cm} |>{\centering\arraybackslash}m{2.3cm} |>{\centering\arraybackslash}m{3.8cm} |}
\hline
& This Work & TMS320C6678 & Kintex-7 480T & NVIDIA Jetson Orin NX 16GB\\
\hline
Energy Efficiency [nJ/b] & 0.158 & 1791.6 & 12.036 & 38.770\\
\hline
Throughput [Gb/s] & 47.8 & 0.0439 & 0.6174 & 0.6448 \\
\hline
\end{tabular}
\end{table*}

\subsection{Detection Accuracy}
\subsubsection{Effect of Modulation Order} 

\begin{figure*}[!t]
\centering
\subfloat[BPSK]{\includegraphics[width=0.333\textwidth]{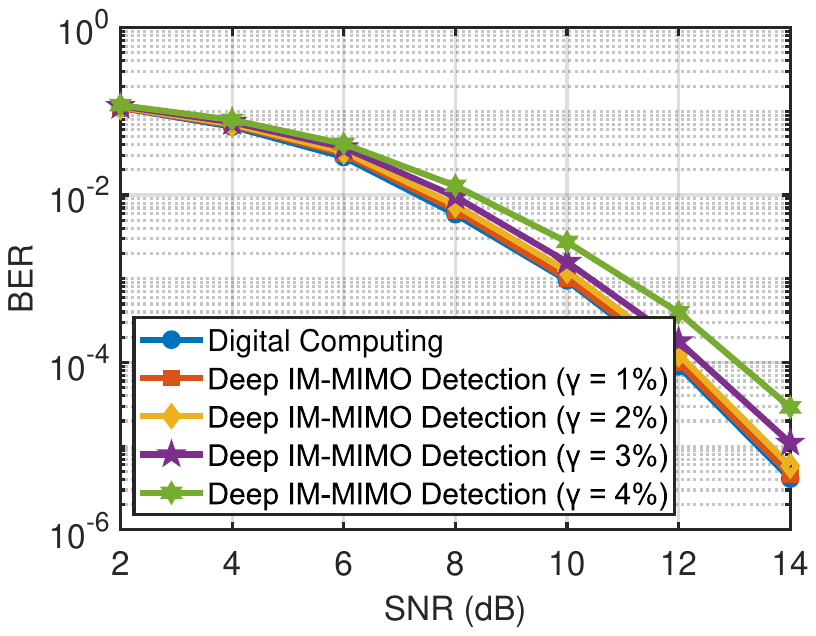}
\label{BPSK}}
\hspace{-3mm}
\subfloat[QPSK]{\includegraphics[width=0.333\textwidth]{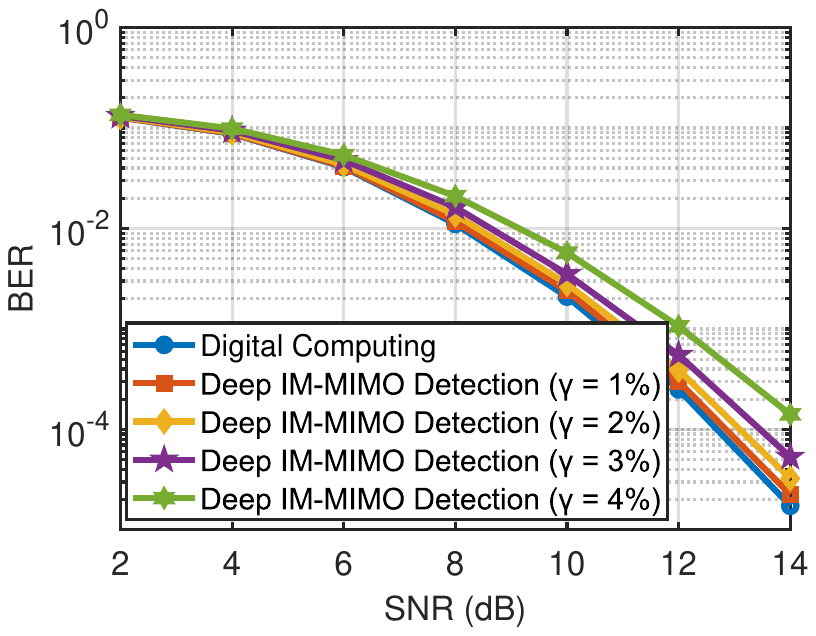}
\label{QPSK}}
\hspace{-3mm}
\subfloat[16QAM]{\includegraphics[width=0.333\textwidth]{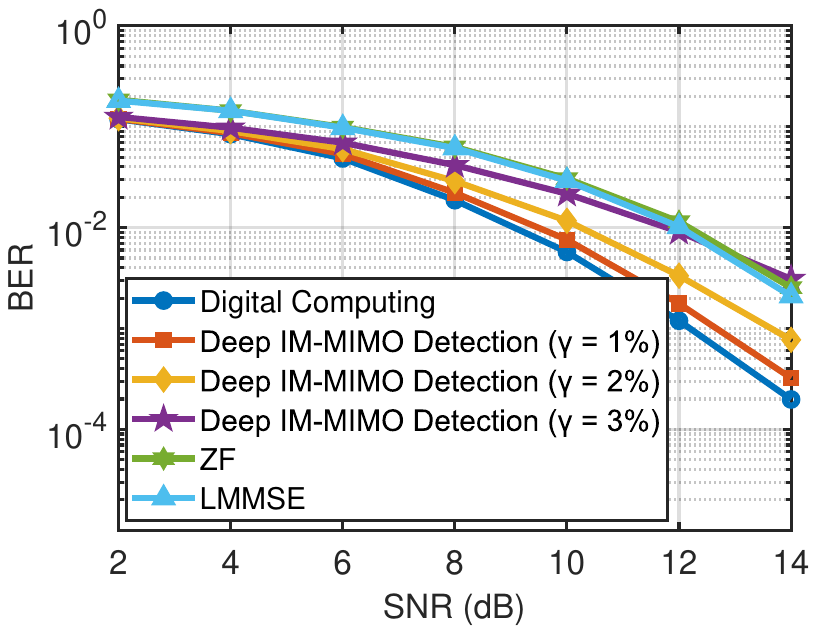}
\label{QAM}}
\caption{BER performance comparisons with the corresponding digital computing counterpart under different modulation orders in a 20 $\times$ 30 Rayleigh fading channel.}
\label{performance_order}
\end{figure*}

As illustrated in Fig. \ref{performance_order}, the BER performance of the proposed deep IM-MIMO detector is evaluated under different C2C variation levels (1\%, 2\%, 3\%, and 4\%), in comparison with the digital computing baseline (without programming noise). It can be observed that when the C2C variation is below 1\%, our design achieves BER performance closely approaching the digital baseline. As the C2C variation increases to 3\%, although there is a trend of BER degradation, our design still demonstrates remarkable robustness, except for the case of $\gamma=3\%$ in 16QAM.

\subsubsection{Effect of Channel Types} 

\begin{figure*}[!t]
\centering
\subfloat[Rayleigh Channel]{\includegraphics[width=0.333\textwidth]{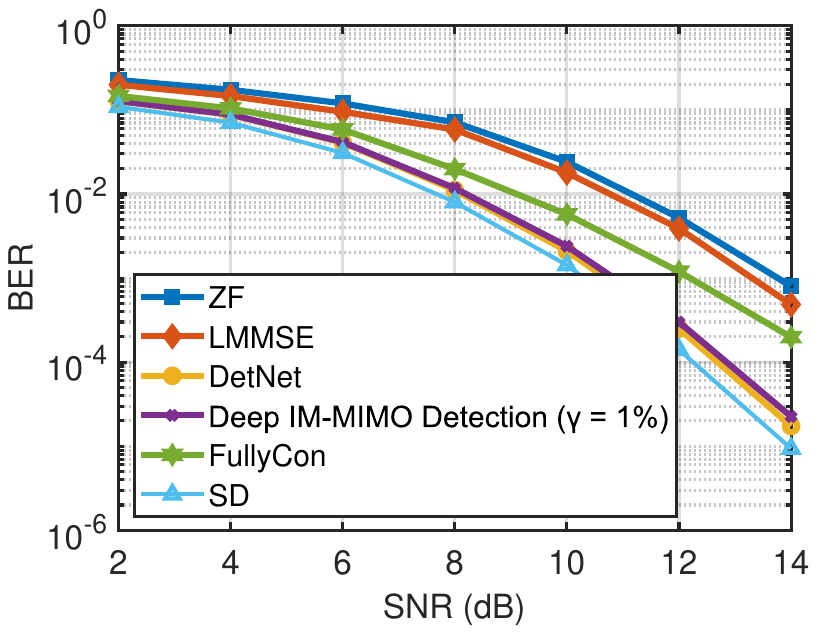}
\label{rayleigh}}
\hspace{-3mm}
\subfloat[Rician Channel]{\includegraphics[width=0.333\textwidth]{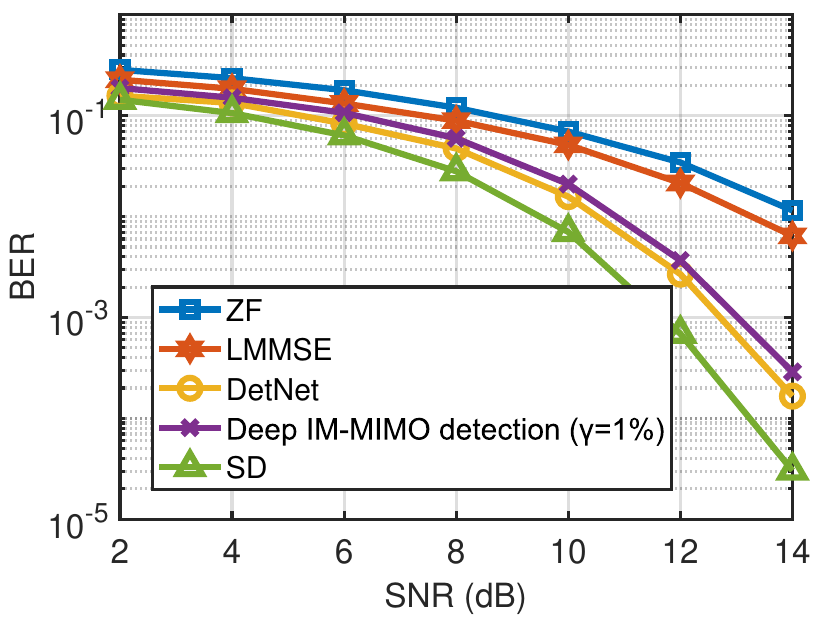}
\label{rician}}
\hspace{-3mm}
\subfloat[Spatially Correlated Channel]{\includegraphics[width=0.333\textwidth]{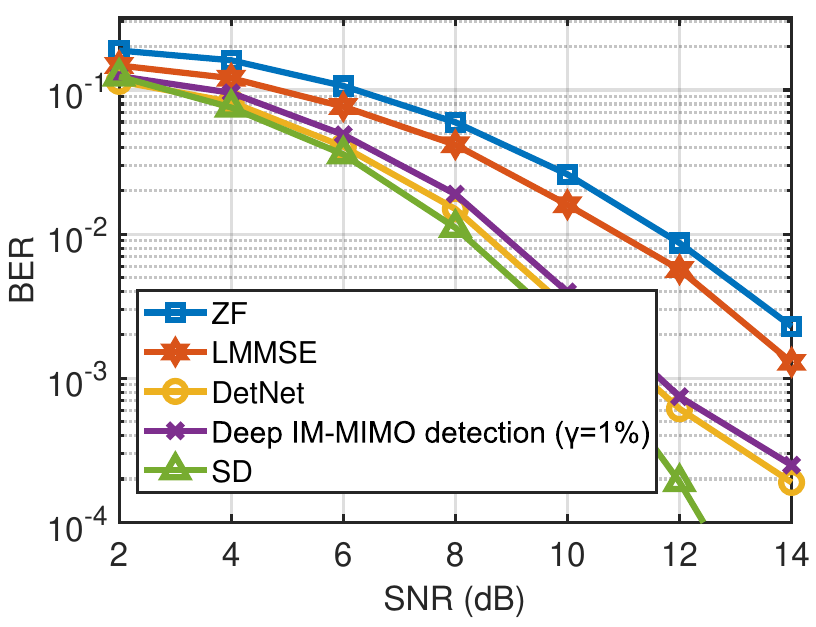}
\label{correlated}}
\caption{BER comparisons under different channels, $(N_t,N_r)=(20,30)$.}
\label{channel_type}
\end{figure*}

Fig. \ref{channel_type} compares the BER performance under different channel types under QPSK modulation. As illustrated in Fig. \ref{rayleigh}, all baseline methods (ZF, MMSE, SD, FullyCon~\cite{samuel2019learning}, and DetNet) are implemented using conventional digital computing platforms. The FullyCon detector is a purely data‑driven model. Linear detectors (ZF and MMSE) offer lower computational complexity, but they suffer from notable degradation in detection accuracy. 
In contrast, the proposed deep IM-MIMO detector not only achieves superior detection accuracy, comparable to that of the optimal SD, but also offers significantly lower processing latency due to its IMC-accelerated architecture. Moreover, the accuracy of the FullyCon detector degrades significantly under time-varying Rayleigh fading channels. In Fig. \ref{rician} and Fig. \ref{correlated}, we include both Rician fading to simulate Line-of-Sight (LoS) conditions (with a Rician K-factor of 3) and spatially correlated Rayleigh fading to simulate realistic Non-Line-of-Sight (NLoS) conditions (with a correlation coefficient of 0.5)~\cite{kermoal2002stochastic}. The BER performance of the proposed design degrades under these more realistic channel conditions but still outperforms conventional linear detectors.

\section{Conclusions}\label{conclusion}
This work presents an ultra-low-latency deep IM-MIMO detector that combines the strengths of deep unfolding and IMC. By transforming an iterative detection algorithm into a feedforward neural network through deep unfolding, the design enables efficient hardware implementation. To accelerate the computationally intensive MVMs, the proposed architecture integrates channel-dependent and neural network modules using memristor-based IMC. To mitigate the impact of noise originating from both the IMC hardware and the wireless channel, a noise-aware training scheme is introduced, enhancing the robustness of the detector. The performance of our design is quantified in terms of detection accuracy, processing latency, and hardware complexity, offering insights for practical system design.

By rethinking the interaction between deep unfolding and IMC, this work provides both theoretical insights and practical guidance for AI-native baseband processing in 6G networks, thereby laying a foundation for future research. One promising direction involves extending the integration of IMC and deep unfolding to address the channel estimation problem, where computational efficiency is often limited by intensive MVMs. Additionally, exploring the synergy between AI and IMC for other wireless communication challenges, such as hybrid beamforming in terahertz systems, physical layer security, and resource management for 6G multi-band networks, represents a rich area for follow-up investigation. As a preliminary study, this work is expected to stimulate more innovations in the development of IMC-based wireless modules for future ultra-low-latency wireless communications.

{\appendices
\section{Proof of Lemma 1}
Consider the crossbar arrays are fully reset. A train of programming pulses is applied to increase the conductance from $G_{\rm off}$ to the target conductance $G$. Consequently, the difference $\Delta G$ between $G$ and the resultant imprecise conductance state follows the distribution \cite{zeng2023realizing}:
    \begin{equation}
        \Delta G| G\sim \mathcal{N}\left( {0,\frac{{{\sigma _{{\Delta}g}^2}(G - {G_{\rm off}}){N_p}}}{{{G_{\rm on}} - {G_{\rm off}}}}} \right). \label{delta_G}
    \end{equation}

For $h_{ij}>0$, the error $\Delta h_{ij}=\frac{G_{ij}^++\Delta G_{ij}-G_{\mathrm{off}}}{\mu}-\frac{G_{ij}^+-G_{\mathrm{off}}}{\mu}=\frac{\Delta G_{ij}}{\mu}$; for $h_{ij}<0$, $\Delta h_{ij}=\frac{G_{\mathrm{off}}-(G_{ij}^-+\Delta G_{ij})}{\mu}-\frac{G_{\mathrm{off}}-G_{ij}^-}{\mu}=-\frac{\Delta G_{ij}}{\mu}$. Combining both cases, the conditional variance of $\Delta h_{ij}$ given $G_{ij}$ is derived as 
\begin{equation} \label{variance}
\begin{split}
    \mathrm{Var}[\Delta h_{ij}|G_{ij}] &= \frac{1}{\mu^2}\mathrm{Var}[\Delta G_{ij}| G_{ij}] \\
    &= \frac{\sigma_{\Delta g}^2 (G_{ij} - G_{\mathrm{off}}) N_p}{\mu^2 (G_{\mathrm{on}} - G_{\mathrm{off}})}.
\end{split}
\end{equation}

Substituting $G_{ij}-G_{\mathrm{off}}=\mu |h_{ij}|$ and $\sigma_{\Delta g}=\gamma(G_{\mathrm{on}}-G_{\mathrm{off}})$ into (\ref{variance}), the conditional variance of $\Delta h_{ij}$ conditioned on $h_{ij}$ yields 
\begin{equation}\label{delta_h}
    \sigma _{\Delta {h_{ij}}}^2 = 3{\gamma ^2}{N_p}|{h_{ij}}|.
\end{equation}

This completes the proof. $\hfill\blacksquare$

\section{Proof of Lemma 2}
Based on the following equivalent transformations
\begin{equation}
    \mathbf{u}_k=\mathbf{C}_1\mathbf{q}_{k-1}+\mathbf{c}_2,
\end{equation}
\begin{equation}
{ \Delta\mathbf{C}_1} = \left[ {\begin{array}{*{20}{c}}
{{\alpha _{2k}}\mathbf{({H^\textit{T}}} \Delta \mathbf{H} + \Delta \mathbf{{H^\textit{T}}H})}&\mathbf{0}\\
\mathbf{0}&\mathbf{0}
\end{array}} \right],\label{C1}
\end{equation}
\begin{equation}
    \Delta {\mathbf{c}_2} = \left[ {\begin{array}{*{20}{c}}
{ - {\alpha _{1k}}\mathbf{({H^\textit{T}}}\mathbf{n} + \Delta \mathbf{H^\textit{T}y_{\mathrm{0}})}}\\
\mathbf{0}
\end{array}} \right], \label{C2}
\end{equation}
we can derive the relationship between ${\mathbf{q}_k}$ and ${\mathbf{q}_{k+1}}$ as follows
\begin{equation}
\begin{aligned}
\mathbf{q}_{k} &= \mathbf{\hat{W}}_{2k}\mathbf{z}_k \\
&= \mathbf{\hat{W}}_{2k}\mathbf{D}_k\left( \mathbf{W}_{1k}\mathbf{u}_k \right)\\
&= \mathbf{\hat{W}}_{2k}\mathbf{D}_k\left( \mathbf{W}_{1k}\mathbf{C}_1\mathbf{q}_{k-1} + \mathbf{W}_{1k}\mathbf{c}_2\right)\\
&= \mathbf{\hat{W}}_{2k}\mathbf{D}_k\mathbf{W}_{1k}\mathbf{C}_1\mathbf{q}_{k-1} + \mathbf{\hat{W}}_{2k}\mathbf{D}_k\mathbf{W}_{1k}\mathbf{c}_2
\end{aligned}\label{hat_u_k}
\end{equation}

When additive noise is present, the input becomes $\mathbf{q}_{k-1}+\mathbf{e}_{k-1}$ in the $k$-th block. Consequently, the output error of the $k$-th block is $\mathbf{e}_{k}=\mathbf{\hat{W}}_{2k}\mathbf{D}_k\mathbf{W}_{1k}(\mathbf{C}_1+\Delta \mathbf{C}_1)\mathbf{e}_{k-1}+\mathbf{\hat{W}}_{2k}\mathbf{D}_k\mathbf{W}_{1k}\Delta \mathbf{C}_1\mathbf{q}_{k-1}+\mathbf{\hat{W}}_{2k}\mathbf{D}_k\mathbf{W}_{1k}\Delta \mathbf{c}_2$. Taking the spectral norm expectation on both sides, we have 
\begin{equation}\label{e_k}
\begin{aligned}
    \mathbb{E}{\left\| {{\mathbf{e}_{k}}} \right\|_2} \le {} &\mathbb{E}{\left\| {{\mathbf{\hat{W}}_{2k}}{\mathbf{D}_k}{\mathbf{W}_{1k}}}{(\mathbf{C}_1+\Delta\mathbf{C}_1)} \right\|_2} \cdot \mathbb{E}{\left\| {{\mathbf{e}_{k-1}}} \right\|_2} \\
    &+ \mathbb{E}{\left\| {{\mathbf{\hat{W}}_{2k}}{\mathbf{D}_k}{\mathbf{W}_{1k}}}{\Delta\mathbf{C}_1} \right\|_2} \cdot \mathbb{E}{\left\| {{\mathbf{q}_{k-1}}} \right\|_2} \\
    &+ \mathbb{E}{\left\| {{\mathbf{\hat{W}}_{2k}}{\mathbf{D}_k}{\mathbf{W}_{1k}}}\Delta {\mathbf{c}_2} \right\|_2}.
\end{aligned}
\end{equation}

This completes the proof. $\hfill\blacksquare$

\section{Proof of Lemma 3}
Based on (\ref{delta_h}), $\Delta {h_{ij}}$ follows the conditional distribution $\Delta {h_{ij}}|{h_{ij}}\sim{\cal N}\left( {0,\sigma _{\Delta {h_{ij}}}^2} \right)$. By the law of total variance and the law of total expectation, we have
\begin{equation}
\text{Var} \left[ \Delta h_{ij} \right] = \mathbb{E}_h \left[ \text{Var} \left( \Delta h_{ij} \mid h_{ij} \right) \right] = 3{\gamma ^2}{N_p} \mathbb{E} \left[ |h_{ij}| \right],
\end{equation}
\begin{equation}
    \mathbb{E}[\Delta h_{ij}]=\mathbb{E}_h[\mathbb{E}[\Delta h_{ij}|h_{ij}]]=0.
\end{equation}
Since $h_{ij}\sim \mathcal{N} (0,1)$, $\mathbb{E}[|h_{ij}|]=\sqrt {\frac{2}{\pi}}$. As a result, the unconditional variance of $\Delta h_{ij}$ is 
\begin{equation}
    \sigma_{\Delta h}^2=3\sqrt {\frac{2}{\pi }} {\gamma ^2}{N_p}.
\end{equation}
According to a fundamental inequality $\|\Delta\mathbf{H}\|_2 \le \|\Delta\mathbf{H}\|_F$, we have $\mathbb{E}\|\Delta\mathbf{H}\|_2 \le \mathbb{E}\|\Delta\mathbf{H}\|_F$, where $\|\Delta\mathbf{H}\|_F = \sqrt{{\textstyle \sum_{i=1}^{2N_t}} {\textstyle \sum_{j=1}^{2N_r}}(\Delta h_{ij})^2 }$ denotes the Frobenius norm of $\Delta\mathbf{H}$. Using the linearity of expectation, $\mathbb{E}[\|\Delta \mathbf{H}\|_F^2]$ is derived as
\begin{equation}
\begin{split}
    \mathbb{E}[\|\Delta \mathbf{H}\|_F^2] &=\sum_{i=1}^{2N_t}\sum_{j=1}^{2N_r}\mathbb{E}[(\Delta h_{ij})^2] \\
    &= 12\sqrt {\frac{2}{\pi }} {\gamma ^2}{N_p}N_tN_r.
\end{split}
\end{equation}
By Jensen’s inequality, $\mathbb{E}[\|\Delta \mathbf{H}\|_F] \le \sqrt{\mathbb{E}[\|\Delta \mathbf{H}\|_F^2]}$. Thus, 
\begin{equation}
    \mathbb{E}[\|\Delta \mathbf{H}\|_2] \le \mathbb{E}[\|\Delta \mathbf{H}\|_F] \le 2\gamma\sqrt{3\sqrt{\frac{2}{\pi}}N_pN_tN_r} .
\end{equation}

Given that $\mathbf{H}\in\mathbb{R} ^{2N_r \times 2N_t}$ is an i.i.d. Gaussian matrix with entries following $\mathcal{N}(0, 1)$, the upper bound of $\mathbb{E}{\left\| \mathbf{H} \right\|_2}$ is given by \cite{vershynin2010introduction}
\begin{equation}
    \mathbb{E}{\left\| \mathbf{H} \right\|_2} \le \sqrt {{2N_t}}  + \sqrt {{2N_r}}.
\end{equation}

During inference, ${\alpha _{ik}}$ is a constant. Define $\varpi_i = \max \alpha _{ik}$ for $i \in \{1, 2\}$ and $k\in \{1,\cdots ,L\}$. We thus derive the following upper bound
\begin{equation}
\begin{aligned}
    &\mathbb{E}\left\| \mathbf{C}_1+\Delta \mathbf{C}_1 \right\|_2 \\
    &\leq \varpi_2 \mathbb{E}\left\| \mathbf{H^\textit{T} H + H^\textit{T} }\Delta \mathbf{H} + \Delta \mathbf{H}^\textit{T} \mathbf{H} \right\|_2 \\
    &\leq \varpi_2 (2 \mathbb{E}\left\| \mathbf{H}^\textit{T} \Delta \mathbf{H} \right\|_2 + \mathbb{E}\left\| \mathbf{H}^\textit{T} \mathbf{H} \right\|_2) \\
    &\leq \varpi_2 (2 \mathbb{E}\left\| \mathbf{H} \right\|_2 \mathbb{E}\left\| \Delta \mathbf{H} \right\|_2 + \mathbb{E}\left\| \mathbf{H} \right\|_2 \mathbb{E}\left\| \mathbf{H} \right\|_2) \\
    &\leq \varpi_2  \left (\gamma\sqrt{6\sqrt{\frac{2}{\pi}}N_p}(\sqrt{N_t}+\sqrt{N_r})^3+2(\sqrt{N_t}+\sqrt{N_r})^2\right ).
\end{aligned}
\end{equation}

This completes the proof. $\hfill\blacksquare$

\section{Proof of Lemma 4} 
The expectation of the spectral norm of ${\alpha _{1k}}({\mathbf{H}^T}\mathbf{n} + \Delta {\mathbf{H}^T}\mathbf{y}_0)$ in (\ref{C2}) is derived as follows:
\begin{equation}
\begin{aligned}
&\mathbb{E}\left\| \mathbf{H}^T \mathbf{n} + \Delta \mathbf{H}^T \mathbf{y}_0 \right\|_2 \\
&\le \mathbb{E}\left\| \mathbf{H}^T \mathbf{n} \right\|_2 + \mathbb{E}\left\| \Delta \mathbf{H}^T \mathbf{H} \mathbf{x} \right\|_2 \\
&\le \mathbb{E}\left\| \mathbf{H} \right\|_2 \mathbb{E}\left\| \mathbf{n} \right\|_2 + \mathbb{E}\left\| \Delta \mathbf{H} \right\|_2 \mathbb{E}\left\| \mathbf{H} \right\|_2 \mathbb{E}\left\| \mathbf{x} \right\|_2 \\
&\le \mathbb{E}\left\| \mathbf{H} \right\|_2 \mathbb{E}\left\| \mathbf{n} \right\|_2 + \mathbb{E}\left\| \Delta \mathbf{H} \right\|_2 \mathbb{E}\left\| \mathbf{H} \right\|_2 
\end{aligned}
\end{equation}
Given that ${\left\| \mathbf{n} \right\|_2^2}$ follows a chi-square distribution $\left\| \mathbf{n} \right\|_2^2\sim\sigma _n^2{\chi ^2}({2N_r})$, we have 
\begin{equation}
    \mathbb{E}\left\| \mathbf{n} \right\|_2^2 = 2{N_r}\sigma _n^2.
\end{equation}
Considering the square root function is a concave function and applying Jensen's inequality, the expectation of $\left\| \mathbf{n} \right\|_2$ satisfies 
\begin{equation}
    \mathbb{E}\left\| \mathbf{n} \right\|_2^{} \le \sqrt {\mathbb{E}\left\| \mathbf{n} \right\|_2^2}  = \sigma _n^{}\sqrt {2{N_r}}.
\end{equation}
As a result, we obtain the following bound for (\ref{C2})
\begin{equation}
\begin{aligned}
    &\mathbb{E}{\left\| {{\Delta \mathbf{c}_2}} \right\|_2}\\
     &\le \varpi_1 \left (\mathbb{E}\left\| \mathbf{H} \right\|_2 \mathbb{E}\left\| \mathbf{n} \right\|_2 + \mathbb{E}\left\| \Delta \mathbf{H} \right\|_2 \mathbb{E}\left\| \mathbf{H} \right\|_2 \right ) \\
    &\le \varpi_1(\sqrt{2N_t}+\sqrt{2N_r})\left (\sigma _n\sqrt{2N_r}+2\gamma\sqrt{3\sqrt{\frac{2}{\pi}}N_pN_tN_r}\right )\\
    & \le 2\varpi_1\left ( \sigma_n\sqrt{N_r(N_t + N_r)} + \gamma\sqrt{3\sqrt{\frac{2}{\pi}}N_p} \cdot (N_t + N_r)^{\frac{3}{2}} \right )
\end{aligned}
\end{equation}

This completes the proof. $\hfill\blacksquare$

\section{Proof of Theorem 1}
For noise-free input $\mathbf{q}_k$, we have 
\begin{equation}
\begin{split}
    \mathbb{E}\left\| \mathbf{q}_{k} \right\|_2 &\le \mathbb{E}\left\| {{\mathbf{\hat W}_{2k}}{\mathbf{D}_k}{\mathbf{W}_{1k}}\mathbf{C}_1} \right\|_2 \mathbb{E}\left\| \mathbf{q}_{k-1} \right\|_2 \\
    &+ \mathbb{E}\left\| {{\mathbf{\hat W}_{2k}}{\mathbf{D}_k}{\mathbf{W}_{1k}}\mathbf{c}_2} \right\|_2.
   \label{noise_free_q}
\end{split}
\end{equation}
Define $\mathbb{E}\left\| {{\mathbf{\hat W}_{2k}}{\mathbf{D}_k}{\mathbf{W}_{1k}}\mathbf{C}_1} \right\|_2 \le {\bar a}$ and $\mathbb{E}{\left\| {{\mathbf{W}_{2k}}{\mathbf{D}_k}{\mathbf{W}_{1k}}{(\mathbf{C}_1+\Delta\mathbf{C}_1)}} \right\|_2}\le {\bar b} $. Since (\ref{noise_free_q}) is a linear non-homogeneous recursion with $\mathbf{q}_0=\mathbf{0}$, the general solution under this condition is:
\begin{equation}
\mathbb{E}{\left\| \mathbf{q}_{k} \right\|_2 \le \frac{{\bar a}^{k}-1}{{\bar a}-1}\mathbb{E}\left\| {{\mathbf{\hat W}_{2k}}{\mathbf{D}_k}{\mathbf{W}_{1k}}\mathbf{c}_2} \right\|_2}.
\label{qk}
\end{equation}
From the recursive relation in (\ref{e_k}), we have 
\begin{equation}
\begin{split}
    {\mathbb E}{{\left\| {{{\bf{e}}_L}} \right\|}_2} &\le \mathbb{E}\left\| {{\mathbf{\hat W}_{2k}}{\mathbf{D}_k}{\mathbf{W}_{1k}}\Delta \mathbf{c}_2} \right\|_2 \cdot \sum\limits_{i = 0}^{L - 1} {{{\bar b}^i}} \\
    &+ \mathbb{E}\left\| {{\mathbf{\hat W}_{2k}}{\mathbf{D}_k}{\mathbf{W}_{1k}}\Delta \mathbf{C}_1} \right\|_2 \cdot \sum\limits_{k = 1}^L {{{\bar b}^{L - k}}{\mathbb E}{{\left\| {{{\bf{q}}_k}} \right\|}_2}},
    \label{e_L2}
\end{split}
\end{equation}
where ${\mathbb E}{{\left\| {{{\bf{e}}_0}} \right\|}_2}=0$. Substituting (\ref{qk}) into (\ref{e_L2}) and rearranging terms, we obtain
\begin{equation}\label{bound_eL}
\begin{split}
{\mathbb E}{{\left\| {{{\bf{e}}_L}} \right\|}_2} &\le \frac{{1 - {{\bar b}^{L}}}}{{1 - {\bar b}}}{\mathbb E}{\left\| {{{{\bf{\hat W}}}_{2k}}{{\bf{D}}_k}{{\bf{W}}_{1k}}\Delta {{\bf{c}}_2}} \right\|_2} \\
& + \frac{{{{\bar a}^L}-{\bar b}^L}}{{({\bar a} - {\bar b})({\bar a} - 1)}} {{\mathbb E}{\left\| {{{{\bf{\hat W}}}_{2k}}{{\bf{D}}_k}{{\bf{W}}_{1k}}\Delta {{\bf{C}}_1}} \right\|}_2}\\
&\quad\times {\mathbb E}{{\left\| {{{{\bf{\hat W}}}_{2k}}{{\bf{D}}_k}{{\bf{W}}_{1k}}{{\bf{c}}_2}} \right\|}_2}.
\end{split}
\end{equation}

For ${{\mathbf{\hat W}_{2k}}{\mathbf{D}_k}{\mathbf{W}_{1k}}}$ in (\ref{bound_eL}), it quantifies the input transformation properties from a neural network perspective. Leveraging the network initialization strategy in \cite{allen2019convergence}, each computational block in our system can be modeled as a two-layer fully-connected neural network. By applying the result in \cite{zhu2022robustness}, the following spectral norm bound holds
\begin{equation}
\mathbb{E}\left\|\hat{\mathbf{W}}_{2k}\mathbf{D}_k\mathbf{W}_{1k}\right\|_2 \leq \underbrace{\sqrt{\frac{4S}{3N_r}}e^{-S/8}}_{\Omega}+1.
\end{equation}

Based on Lemmas 2, 3, and 4, the upper bound of ${\mathbb E}{{\left\| {{{\bf{e}}_L}} \right\|}_2}$ is derived as follows
\begin{equation}
\begin{split}
    \mathbb{E}\|\mathbf{e}_L\|_2 &\leq \frac{1 - \tau{(\Omega+1)}^L}{1 - \tau{(\Omega+1)}} \xi  {(\Omega+1)} \\
    &+ \frac{C\Phi^5(\varphi^L-\tau^L)}{(\varphi{(\Omega+1)} - 1)(\varphi-\tau)} \cdot{(\Omega+1)}^{L+1},\label{appendix_bound}
\end{split}
\end{equation}
where 
\begin{equation}
    C = 2\varpi_1 \varpi_2  \gamma \sqrt{6\sqrt{\frac{2}{\pi}}N_p},
\end{equation}
\begin{equation}
    \tau = \varpi_2  \left( \gamma \sqrt{6 \sqrt{ \frac{2}{\pi}} N_p} \cdot \Phi^3 + 2 \Phi^2 \right),
\end{equation}
\begin{equation}
    \varphi = 2\varpi_2  \Phi^2.
\end{equation}
After performing a Taylor expansion on the right-hand side of (\ref{appendix_bound}), we denote 
\begin{equation}\label{secondTerm}
\begin{split}
    &\frac{C\Phi^5(\varphi^L-\tau^L)}{(\varphi{(\Omega+1)} - 1)(\varphi-\tau)} \cdot{(\Omega+1)}^{L+1}\\
    &= \frac{C\Phi^5(\varphi^L-\tau^L)}{(\varphi - 1)(\varphi-\tau)}+\frac{C\Phi^5(\varphi^L-\tau^L)}{(\varphi - 1)(\varphi-\tau)}\cdot \left ( L-\frac{1}{\varphi-1} \right ){\Omega}\\
    &\quad+O(\Omega^2),
\end{split}
\end{equation}
and
\begin{equation}\label{firstTerm}
    \frac{1 - \tau{(\Omega+1)}^L}{1 - \tau{(\Omega+1)}} \xi  {(\Omega+1)}= \xi+\xi\cdot \frac{1-\tau L}{1-\tau}{\Omega}+O(\Omega^2).
\end{equation}
Here, $\Gamma=\frac{C\Phi^5(\varphi^L-\tau^L)}{(\varphi - 1)(\varphi-\tau)}$ characterizes the cumulative coefficient of output error over $L$ blocks. 

Substituting (\ref{secondTerm}) and (\ref{firstTerm}) into (\ref{appendix_bound}), the bound can be rewritten as
\begin{equation}
\begin{split}
    \mathbb{E}\|\mathbf{e}_L\|_2 &\leq \xi + \Gamma+\left [ \xi\cdot \frac{1-\tau L}{1-\tau}+ \Gamma\cdot \left( L-\frac{1}{\varphi-1} \right)\right ]{\Omega}\\
    &+O(\Omega^2).
\end{split}
\end{equation}

This completes the proof. $\hfill\blacksquare$}

\bibliography{Reference}
\bibliographystyle{IEEEtran}

\end{document}